\documentclass[12pt]{amsart}

\usepackage{complexity}
\usepackage{xspace}
\usepackage{tikz}
\usetikzlibrary{calc}
\usepackage{frenchineq}
\usepackage{amsmath}  
\usepackage{amssymb}     
\usepackage{amsthm}  
\usepackage{bbm}    
\usepackage{bm}
\usepackage{accents}
\usepackage{mathrsfs}
\usepackage{fullpage}
\usepackage{microtype}
\usepackage[colorlinks,linkcolor=blue,citecolor=magenta]{hyperref}
\usepackage[colorinlistoftodos,bordercolor=orange,backgroundcolor=orange!20,linecolor=orange,textsize=scriptsize]{todonotes}
\usepackage{soul}
\usepackage{yhmath}
\usepackage{subfig}

\newcommand{\QK}{\textup{\textsc{Quasi-Kernel}}\xspace}
\newcommand{\MinQK}{\textup{\textsc{Min-Quasi-Kernel}}\xspace}

\DeclareMathOperator{\opt}{opt}

\newtheorem{theorem}{Theorem}[section]
\newtheorem{proposition}[theorem]{Proposition}

\newtheorem{lemma}[theorem]{Lemma}

{\bfseries}{}

\theoremstyle{remark}
\newtheorem{claim}{Claim}

\title{Algorithmic aspects of quasi-kernels}

\author{Hélène Langlois \and Frédéric Meunier \and Romeo Rizzi \and Stéphane Vialette}

\address[Hélène Langlois]{CERMICS, École des Ponts ParisTech, 77455 Marne-la-Vallée, France \and LIGM, Univ Gustave Eiffel, 77454 Marne-la-Vallée, France}
\email{helene.langlois@enpc.fr}

\address[Frédéric Meunier]{CERMICS, École des Ponts ParisTech, 77455 Marne-la-Vallée, France}
\email{frederic.meunier@enpc.fr}

\address[Romeo Rizzi]{Department of Computer Science, Università di Verona,  37129 Verona, Italy}
\email{romeo.rizzi@univr.it}

\address[Stéphane Vialette]{LIGM, Univ Gustave Eiffel, CNRS, 77454 Marne-la-Vallée, France}
\email{stephane.vialette@univ-eiffel.fr}

\begin{document}

\begin{abstract}
In a digraph, a quasi-kernel is a subset of vertices that is independent and such that every vertex can reach some vertex in that set via a directed path of length at most two.
	Whereas Chv{\'a}tal and Lov{\'a}sz proved  in 1974 that every digraph has a quasi-kernel,
 	very little is known so far about the complexity of finding small quasi-kernels.
 	In 1976 Erd\H{o}s and Sz\'ekely conjectured that every sink-free digraph $D = (V, A)$ has a quasi-kernel 
	of size at most $|V|/2$.
 	Obviously, if $D$ has two disjoint quasi-kernels then it has a quasi-kernel of size at most $|V|/2$, and in 2001 Gutin, Koh, Tay and Yeo conjectured that every sink-free digraph has two disjoint quasi-kernels. Yet, they constructed in 2004 a counterexample, thereby disproving
this stronger conjecture.
 	We shall show that, not only sink-free digraphs occasionally fail to contain two disjoint quasi-kernels, 
  but it is computationally hard to distinguish those that do from those that do not. 
 	We also prove that the problem of computing a small quasi-kernel is polynomial time solvable for 
 	orientations of trees but is computationally hard in most other cases 
 	(and in particular for restricted acyclic digraphs). 
\end{abstract}


\maketitle

\section{Introduction}
\label{sec:intro}

Let $D = (V, A)$ be a digraph. A \emph{kernel} $K$ is a subset of vertices that is independent 
(\emph{i.e.}, all pairs of distinct vertices of $K$ are non-adjacent) and such that, for every 
vertex $v \notin K$, there exists $w \in K$ with $(v, w) \in A$.
Kernels were introduced by von Neumann and Morgenstern~\cite{vonneumann1947}. 
It is now a central notion in graph theory and has important applications in relations with 
colorings~\cite{galvin_list_1995}, perfect graphs~\cite{DBLP:journals/dm/BorosG06}, game theory and economics~\cite{igarashi_coalition_2017}, logic~\cite{kernel_SAT}, etc.
Clearly, not every digraph has a kernel 
(for instance, a directed cycle of odd length does not contain a kernel) and a digraph may have several
kernels.
Chv{\'a}tal proved that deciding whether a digraph has a kernel is 
\NP-complete~\cite{chvatal_computational_1973} and the problem is equally hard
for planar digraph with bounded degree~\cite{FRAENKEL1981257}.

Chv\'atal and Lov\'asz~\cite{10.1007/BFb0066192} later on introduced the notion of quasi-kernels.
A \emph{quasi-kernel} in a digraph is a subset of vertices that is independent and such that every vertex can reach some vertex in that set via a directed path of length at most two.
Defining the (directed) distance $d(v,w)$ from a vertex $v$ to a vertex $w$ as the minimum length of a directed path from $v$ to $w$, a quasi-kernel $Q$ is a subset of vertices that is independent and such that for every vertex $v\notin Q$ there exists $w\in Q$ such that $d(v,w)\leq 2$. In particular, any kernel is a quasi-kernel.
Yet, unlike kernels, every digraph has a quasi-kernel.  Chvátal and Lovász provided a proof of this fact, which can be turned into an easy polynomial time algorithm (alternate simple proofs exist~\cite{DBLP:journals/jgt/Bondy03}).
However, deciding whether there exists a quasi-kernel that contains a specified vertex is 
\NP-complete~\cite{CROITORU2015863}.
Jacob and Meyniel~\cite{DBLP:journals/dm/JacobM96} proved that if a digraph does not have a kernel 
then it must contain at least three (not necessarily disjoint) quasi-kernels.
Digraphs with exactly one and two quasi‐kernels have been characterized by Gutin et al.~\cite{Gutin_2004}. 
It follows from this characterization that if a digraph has precisely 
two quasi-kernels then these two quasi-kernels are actually disjoint. 
At a more general level, counting quasi-kernels in digraphs is as hard as counting independent 
sets in graphs~\cite{DBLP:conf/approx/DyerGGJ00}.

In 1976 Erd\H{o}s and Sz\'ekely~\cite{conj-qk} conjectured that
a sink-free digraph $D = (V, A)$  (\emph{i.e.}, every vertex of $D$ has positive out-degree) 
has a quasi-kernel of size at most $|V|/2$
(note that if $D$ has two disjoint quasi-kernels then it has a quasi-kernel of size at most $|V|/2$).
This question is known as the \emph{small quasi-kernel conjecture}.
So far, this conjecture is only confirmed for narrow classes of digraphs.
In 2008, Heard and Huang~\cite{heard_disjoint_2008} showed that every $D = (V, A)$ digraph $D$  
has two disjoint quasi-kernels if $D$ is semicomplete multipartite (including tournaments), 
quasi-transitive (including transitive digraphs), or locally semicomplete. 
Very recently, Kostochka at al.~\cite{kostochka_towards_2020} renewed the interest in
the small quasi-kernel conjecture and proved that the conjecture holds
for orientations of $4$-colorable graphs (in particular, for all planar graphs).

In 2011 Gutin et al.~\cite{Gutin_2001} conjectured that every sink-free digraph has two disjoint quasi-kernels
(this stronger conjecture implies the original small quasi-kernel conjecture). 
In 2004, in an update of their paper,  the authors constructed a counterexample with 14 vertices~\cite{Gutin_2004}.
As we shall prove, not only sink-free digraphs occasionally fail to contain two disjoint quasi-kernels, 
but it is actually computationally hard to distinguish those that do from those that do not. 
Note that, whereas the small quasi-kernel conjecture is true for planar sink-free 
digraphs~\cite{kostochka_towards_2020}, no sink-free planar digraph without two disjoint quasi-kernels 
is known so far
(the counterexample constructed by Gutin et al.~\cite{Gutin_2004} does contain a directed $K_7$).
Whether such a planar graph exists is not known but we shall show that deciding whether a sink-free bounded degree planar 
digraph has three disjoint quasi-kernels is \NP-complete.

Surprisingly enough, whereas every digraph has a quasi-kernel, very little is known about the 
algorithmic aspects of minimizing the size of a quasi-kernel.
In this paper, we initiate the study of the problem of finding a quasi-kernel of minimum size problem which we call \MinQK (and we let \QK stand for the related decision problem).
As we shall see soon, the problem is computationally hard even for simple digraph classes,
so that most of our work is devoted to studying the complexity of finding small quasi-kernels in 
acyclic digraphs.

\section{Disjoint quasi-kernels}

Towards proving the small quasi-kernel conjecture, Gutin et al.~\cite{Gutin_2001} conjectured that every 
sink-free digraph has two disjoint quasi-kernel 
(this stronger conjecture implies the original small quasi-kernel conjecture).
Whereas the same authors constructed a nice counterexample with 14 vertices~\cite{Gutin_2004}
(thereby proving that there exist digraphs that contain neither sinks nor a pair of disjoint quasi-kernels),
Heard and Huang~\cite{heard_disjoint_2008} proved that, for several classes of digraphs, 
the condition of containing no sinks guarantees the existence of a pair of disjoint quasi-kernels. 
The classes contain semicomplete multipartite, quasi-transitive, and locally semicomplete digraphs.
We show that, not only sink-free digraphs occasionally fail to contain two disjoint quasi-kernels, but it is
\NP-complete to distinguish those that do from those that do not
(our proof uses the counterexample constructed by Gutin et al.~\cite{Gutin_2004}). 

\begin{theorem}
	\label{theorem:2QK-1}
 	Deciding if a digraph has two disjoint quasi-kernels is \NP-complete, 
	even for digraphs with maximum out-degree six.
\end{theorem}

\begin{proof}

 Given a boolean expression $F$ in conjunctive normal form (CNF) where each clause is the disjunction of at most three distinct 
 literals, \textsc{3-\SAT} asks to decide whether $F$ is satisfiable. 
  We reduce from \textsc{3-\SAT} which is known to be 
 \NP-complete~\cite{Kar72}.
  
  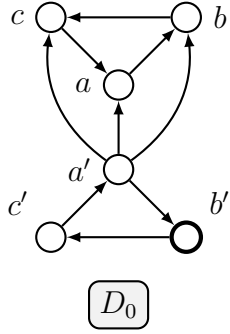
\begin{figure}[t!]
    \centering
    \begin{tikzpicture}
  [
    scale=.9,
    >=latex,
    thick,
    vertex/.style={shape=circle,draw=black},
    qk/.style={vertex,fill=black!25},
    sink/.style={vertex,fill=black!100}
  ]
  \begin{scope}[]
    \node [vertex,label=west:$c$] (C) at (-1,1) {};
    \node [vertex,label=west:$a$] (A) at (0,0) {};
    \node [vertex,label=east:$b$] (B) at (1,1) {};

    \draw [->] (A) to(B);
    \draw [->] (B) to (C);
    \draw [->] (C) to (A);

    \node [vertex,label=north west:$c'$] (Cp) at (-1,-2.25) {};
    \node [vertex,label=west:$a'$] (Ap) at (0,-1.25) {};
    \node [vertex,ultra thick,label=north east:$b'$] (Bp) at (1,-2.25) {};

    \draw [->] (Ap) to (Bp);
    \draw [->] (Bp) to (Cp);
    \draw [->] (Cp) to (Ap);

    \draw [->] (Ap) to (A);
    \draw [->] (Ap) to [bend right=30] (B);
    \draw [->] (Ap) to [bend left=30] (C);

    \node [draw,rounded corners,fill=black!5] (name) at (0,-3.25) {$D_0$};
  \end{scope}

\end{tikzpicture}
    \caption{\label{fig:Disjoint-QK2-0-gadget}%
      Proof of Theorem~\ref{theorem:2QK-1}: gadget.
    }
  \end{figure}

  \begin{figure}[t!]
    \centering
    \begin{tikzpicture}
  [
    scale=.9,
    >=latex,
    thick,
    vertex/.style={shape=circle,draw=black},
    qk/.style={vertex,fill=black!25},
    sink/.style={vertex,fill=black!100}
  ]
  \begin{scope}[]
    \node [vertex,label=north:$c$] (C) at (-1,1) {};
    \node [vertex,label=west:$a$] (A) at (0,0) {};
    \node [vertex,label=north:$b$] (B) at (1,1) {};

    \draw [->] (A) to(B);
    \draw [->] (B) to (C);
    \draw [->] (C) to (A);

    \node [vertex,label=north west:$c'$] (Cp) at (-1,-2.25) {};
    \node [vertex,label=west:$a'$] (Ap) at (0,-1.25) {};
    \node [vertex,ultra thick,label=north east:$b'$] (Bp) at (1,-2.25) {};

    \draw [->] (Ap) to (Bp);
    \draw [->] (Bp) to (Cp);
    \draw [->] (Cp) to (Ap);

    \draw [->] (Ap) to (A);
    \draw [->] (Ap) to [bend right=30] (B);
    \draw [->] (Ap) to [bend left=30] (C);

    \node [draw,rounded corners,fill=black!5] (name) at (0,-3.25) {$D_0$};
  \end{scope}

  \begin{scope}[xshift=2cm,yshift=-4cm]
    \node [vertex,label=west:$A''_i$] (Ci) at (-1,-3.25) {};
    \node [vertex,label=west:$A_i$] (Ai) at (0,-2.25) {};
    \node [vertex,label=east:$A'_i$] (Bi) at (1,-3.25) {};

    \draw [->] (Ai) to(Bi);
    \draw [->] (Bi) to (Ci);
    \draw [->] (Ci) to (Ai);

    \node [vertex,ultra thick,label=south west:$\texttt{f}_i$] (Fi) at (-1,0) {};
    \node [vertex,label=west:$B_i$] (Aip) at (0,-1) {};
    \node [vertex,ultra thick,label=south east:$\texttt{t}_i$] (Ti) at (1,0) {};

    \draw [->] (Aip) to (Ti);
    \draw [->] (Ti) to (Fi);
    \draw [->] (Fi) to (Aip);

    \draw [->] (Aip) to (Ai);
    \draw [->] (Aip) to [bend left=30] (Bi);
    \draw [->] (Aip) to [bend right=30] (Ci);

    \draw [->] (Fi) to (Bp);
    \draw [->] (Ti) to [bend right=30] (Bp);

    \node [draw,rounded corners,fill=black!5] (name) at (0,-4.25) {$D_i$};
  \end{scope}
\end{tikzpicture}
    \caption{\label{fig:Disjoint-QK2-variable-gadget}%
      Proof of Theorem~\ref{theorem:2QK-1}: variable gadget.
    }
  \end{figure}
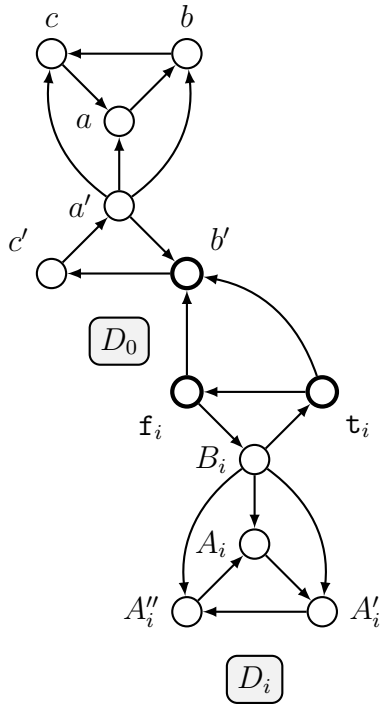

  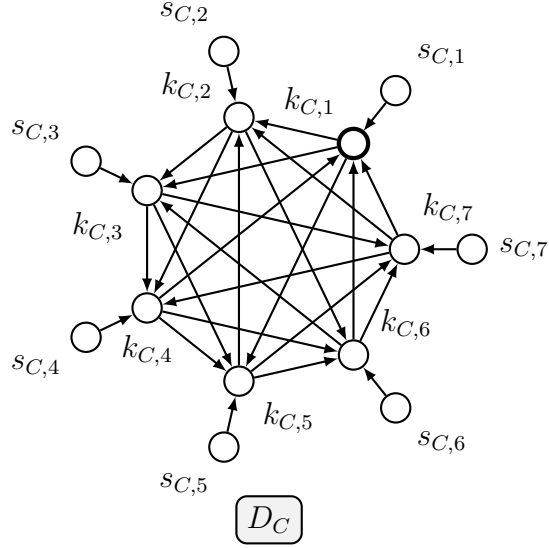
\begin{figure}[t!]
    \centering
    \begin{tikzpicture}
  [
    scale=.9,
    >=latex,
    thick,
    vertex/.style={shape=circle,draw=black},
    qk/.style={vertex,fill=black!25},
    sink/.style={vertex,fill=black!100}
  ]
  \begin{scope}[xshift=0cm,yshift=-13cm]
    \foreach \n in {1,2,...,7} {
      \node [vertex,label=(\n*360/7)+65:$k_{C,\n}$] (K\n) at (\n*360/7:2cm) {};
      \node [vertex,label=\n*360/7:$s_{C,\n}$] (S\n) at (\n*360/7:3cm) {};
      \path [->] (S\n) edge node [left] {} (K\n);
    }
    \node [vertex,ultra thick] (K1) at (360/7:2cm) {};
    \foreach \k in {K2,K3,K5} {
      \path [->] (K1) edge node [left] {} (\k);
    }
    \foreach \k in {K3,K4,K6} {
      \path [->] (K2) edge node [left] {} (\k);
    }
    \foreach \k in {K4,K5,K7} {
      \path [->] (K3) edge node [left] {} (\k);
    }
    \foreach \k in {K5,K6,K1} {
      \path [->] (K4) edge node [left] {} (\k);
    }
    \foreach \k in {K6,K7,K2} {
      \path [->] (K5) edge node [left] {} (\k);
    }
    \foreach \k in {K7,K1,K3} {
      \path [->] (K6) edge node [left] {} (\k);
    }
    \foreach \k in {K1,K2,K4} {
      \path [->] (K7) edge node [left] {} (\k);
    }

    \node [draw,rounded corners,fill=black!5] (name) at (0,-4) {$D_C$};
  \end{scope}

\end{tikzpicture}
    \caption{\label{fig:Disjoint-QK2-clause-gadget}%
      Proof of Theorem~\ref{theorem:2QK-1}: clause gadget.
    }
  \end{figure}

  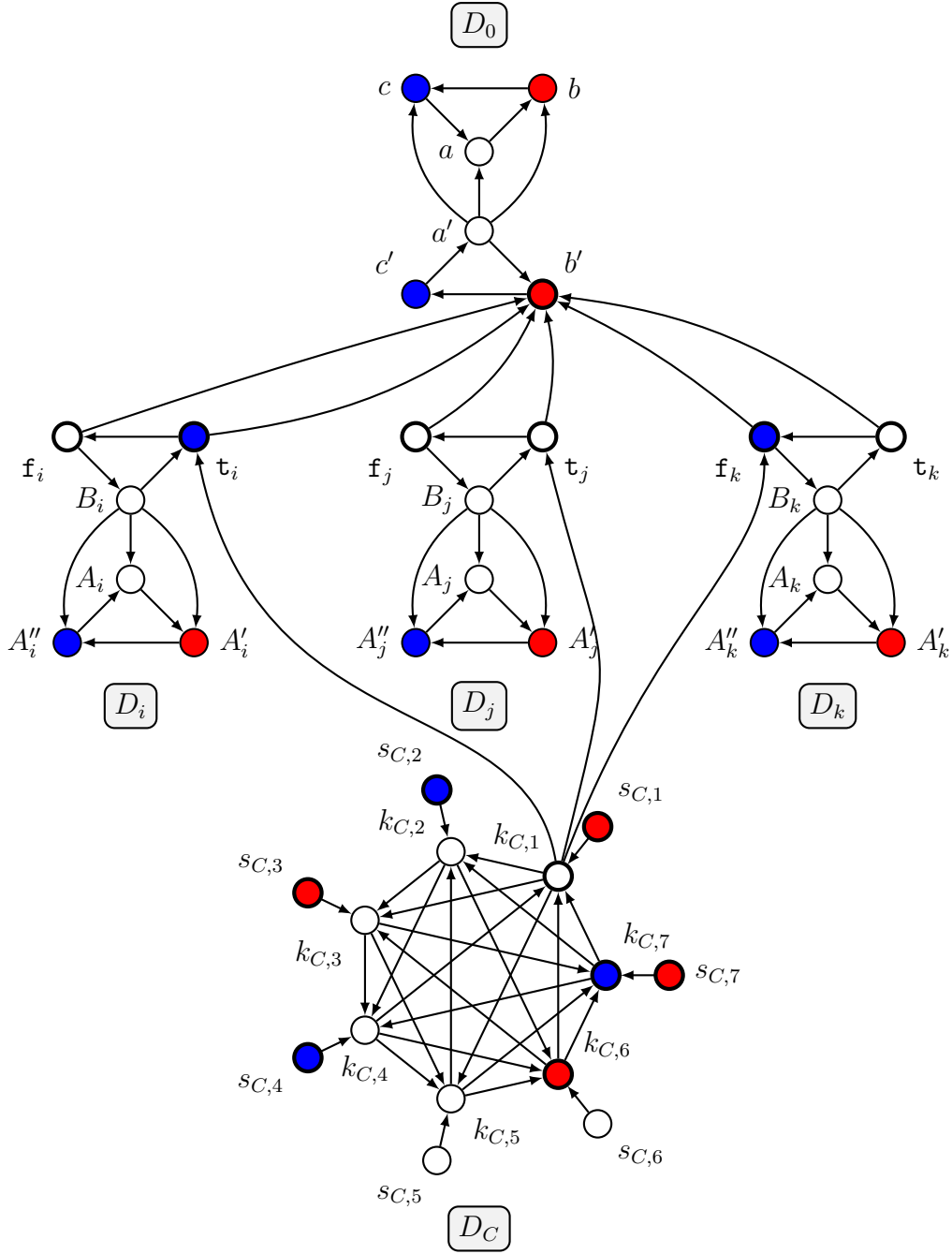
\begin{figure}[t!]
    \centering
    \begin{tikzpicture}
  [
    scale=.9,
    >=latex,
    thick,
    vertex/.style={shape=circle,draw=black},
    qk/.style={vertex,fill=black!25},
    sink/.style={vertex,fill=black!100}
  ]
  \begin{scope}[]
    \node [vertex,fill=blue,label=west:$c$] (C) at (-1,1) {};
    \node [vertex,label=west:$a$] (A) at (0,0) {};
    \node [vertex,fill=red!100,label=east:$b$] (B) at (1,1) {};

    \draw [->] (A) to(B);
    \draw [->] (B) to (C);
    \draw [->] (C) to (A);

    \node [vertex,fill=blue,label=north west:$c'$] (Cp) at (-1,-2.25) {};
    \node [vertex,label=west:$a'$] (Ap) at (0,-1.25) {};
    \node [vertex,fill=red!100,ultra thick,label=north east:$b'$] (Bp) at (1,-2.25) {};

    \draw [->] (Ap) to (Bp);
    \draw [->] (Bp) to (Cp);
    \draw [->] (Cp) to (Ap);

    \draw [->] (Ap) to (A);
    \draw [->] (Ap) to [bend right=30] (B);
    \draw [->] (Ap) to [bend left=30] (C);

    \node [draw,rounded corners,fill=black!5] (name) at (0,2) {$D_0$};
  \end{scope}

  \begin{scope}[xshift=-5.5cm,yshift=-5.5cm]
    \node [vertex,fill=blue,label=west:$A''_i$] (Ci) at (-1,-2.25) {};
    \node [vertex,label=west:$A_i$] (Ai) at (0,-1.25) {};
    \node [vertex,fill=red!100,label=east:$A'_i$] (Bi) at (1,-2.25) {};

    \draw [->] (Ai) to(Bi);
    \draw [->] (Bi) to (Ci);
    \draw [->] (Ci) to (Ai);

    \node [vertex,ultra thick,label=south west:$\texttt{f}_i$] (Fi) at (-1,1) {};
    \node [vertex,label=west:$B_i$] (Aip) at (0,0) {};
    \node [vertex,fill=blue,ultra thick,label=south east:$\texttt{t}_i$] (Ti) at (1,1) {};

    \draw [->] (Aip) to (Ti);
    \draw [->] (Ti) to (Fi);
    \draw [->] (Fi) to (Aip);

    \draw [->] (Aip) to (Ai);
    \draw [->] (Aip) to [bend left=30] (Bi);
    \draw [->] (Aip) to [bend right=30] (Ci);

    \draw [->] (Fi) to [bend left=2] (Bp);
    \draw [->] (Ti) to [bend right=15]  (Bp);

    \node [draw,rounded corners,fill=black!5] (name) at (0,-3.25) {$D_i$};
  \end{scope}

  \begin{scope}[xshift=0cm,yshift=-5.5cm]
    \node [vertex,fill=blue,label=west:$A''_j$] (Cj) at (-1,-2.25) {};
    \node [vertex,label=west:$A_j$] (Aj) at (0,-1.25) {};
    \node [vertex,fill=red!100,label=east:$A'_j$] (Bj) at (1,-2.25) {};

    \draw [->] (Aj) to(Bj);
    \draw [->] (Bj) to (Cj);
    \draw [->] (Cj) to (Aj);

    \node [vertex,ultra thick,label=south west:$\texttt{f}_j$] (Fj) at (-1,1) {};
    \node [vertex,label=west:$B_j$] (Ajp) at (0,0) {};
    \node [vertex,ultra thick,label=south east:$\texttt{t}_j$] (Tj) at (1,1) {};

    \draw [->] (Ajp) to (Tj);
    \draw [->] (Tj) to (Fj);
    \draw [->] (Fj) to (Ajp);

    \draw [->] (Ajp) to (Aj);
    \draw [->] (Ajp) to [bend left=30] (Bj);
    \draw [->] (Ajp) to [bend right=30] (Cj);

    \draw [->] (Fj) to [bend right=16] (Bp);
    \draw [->] (Tj) to [bend right=12] (Bp);

    \node [draw,rounded corners,fill=black!5] (name) at (0,-3.25) {$D_j$};
  \end{scope}

  \begin{scope}[xshift=5.5cm,yshift=-5.5cm]
    \node [vertex,fill=blue,label=west:$A''_k$] (Ck) at (-1,-2.25) {};
    \node [vertex,label=west:$A_k$] (Ak) at (0,-1.25) {};
    \node [vertex,fill=red!100,label=east:$A'_k$] (Bk) at (1,-2.25) {};

    \draw [->] (Ak) to(Bk);
    \draw [->] (Bk) to (Ck);
    \draw [->] (Ck) to (Ak);

    \node [vertex,fill=blue,ultra thick,label=south west:$\texttt{f}_k$] (Fk) at (-1,1) {};
    \node [vertex,label=west:$B_k$] (Akp) at (0,0) {};
    \node [vertex,ultra thick,label=south east:$\texttt{t}_k$] (Tk) at (1,1) {};

    \draw [->] (Akp) to (Tk);
    \draw [->] (Tk) to (Fk);
    \draw [->] (Fk) to (Akp);

    \draw [->] (Akp) to (Ak);
    \draw [->] (Akp) to [bend left=30] (Bk);
    \draw [->] (Akp) to [bend right=30] (Ck);

    \draw [->] (Fk) to [bend right=7] (Bp);
    \draw [->] (Tk) to [bend right=15]  (Bp);

    \node [draw,rounded corners,fill=black!5] (name) at (0,-3.25) {$D_k$};
  \end{scope}

  \begin{scope}[xshift=0cm,yshift=-13cm]
    \foreach \n in {1,2,...,7} {
      \node [vertex,label=(\n*360/7)+65:$k_{C,\n}$] (K\n) at (\n*360/7:2cm) {};
      \node [vertex,label=\n*360/7:$s_{C,\n}$] (S\n) at (\n*360/7:3cm) {};
      \path [->] (S\n) edge node [left] {} (K\n);
    }
    \node [vertex,ultra thick] (K1) at (1*360/7:2cm) {};
    \node [vertex,fill=red!100,ultra thick] (S1) at (1*360/7:3cm) {};
    \node [vertex,fill=red!100,ultra thick] (S7) at (7*360/7:3cm) {};
    \node [vertex,fill=red!100,ultra thick] (S3) at (3*360/7:3cm) {};
    \node [vertex,fill=red!100,ultra thick] (K6) at (6*360/7:2cm) {};
    \node [vertex,fill=blue,ultra thick] (S2) at (2*360/7:3cm) {};
    \node [vertex,fill=blue,ultra thick] (S4) at (4*360/7:3cm) {};
    \node [vertex,fill=blue,ultra thick] (K7) at (7*360/7:2cm) {};

    \foreach \k in {K2,K3,K5} {
      \path [->] (K1) edge node [left] {} (\k);
    }
    \foreach \k in {K3,K4,K6} {
      \path [->] (K2) edge node [left] {} (\k);
    }
    \foreach \k in {K4,K5,K7} {
      \path [->] (K3) edge node [left] {} (\k);
    }
    \foreach \k in {K5,K6,K1} {
      \path [->] (K4) edge node [left] {} (\k);
    }
    \foreach \k in {K6,K7,K2} {
      \path [->] (K5) edge node [left] {} (\k);
    }
    \foreach \k in {K7,K1,K3} {
      \path [->] (K6) edge node [left] {} (\k);
    }
    \foreach \k in {K1,K2,K4} {
      \path [->] (K7) edge node [left] {} (\k);
    }

    \node [draw,rounded corners,fill=black!5] (name) at (0,-4) {$D_C$};
  \end{scope}

  \draw [->] (K1) .. controls ($ (K1) +(-.5,+3) $) and ($ (Ti) +(+1,-5) $) .. (Ti);
  \draw [->] (K1) .. controls ($ (K1) +(1,+4) $) and ($ (Tj) +(+1,-4) $) .. (Tj);
  \draw [->] (K1) .. controls ($ (K1) +(1.5,+4) $) and ($ (Fk) +(0,-2) $) .. (Fk);
\end{tikzpicture}
    \caption{\label{fig:Disjoint-QK2-cnx-gadget}%
      Proof of Theorem~\ref{theorem:2QK-1}: Connecting the gadgets for clause 
      $c = x_i \vee x_j \vee \neg x_k$.
      Red (resp. Blue) vertices denote vertices in $Q_1$ (resp. $Q_2$). 
      Shown here is the case 
      $\varphi(x_i) = \texttt{true}$,
      $\varphi(x_j) = \texttt{false}$ and $\varphi(x_k) = \texttt{false}$
      (\emph{i.e.}, $\texttt{t}_i \in Q_2$, $\texttt{f}_j \in Q_2$ and 
      $\texttt{f}_k \in Q_2$).
      Note that $\texttt{f}_j \notin Q_2$ and $\texttt{t}_j \notin Q_2$
      imply $\varphi(x_j) = \texttt{false}$.
      }
  \end{figure}

  Consider an instance of \textsc{3-\SAT}. 
  Let $X = \{x_1,x_2,\ldots, x_n\}$ be its variables, 
  and let $F = C_1 \vee C_2 \vee \dots \vee C_m$ be its CNF-formula.
  We construct a digraph $D = (V, A)$ as follows.
  \begin{itemize}
    \item  
    We start with the gadget $D_0$ shown in Fig.~\ref{fig:Disjoint-QK2-0-gadget}
    which contains the specified vertex $b'$.
    \item 
    For every boolean variable $x_i \in X$ we introduce the gadget $D_i$ 
    shown in the right part of Fig.~\ref{fig:Disjoint-QK2-variable-gadget} which contains two 
    specified vertices $\texttt{f}_i$ and $\texttt{t}_i$.
    Furthermore, we connect $D_i$ to $D_0$ with two arcs $(\texttt{f}_i, b')$ and $(\texttt{t}_i, b')$.
    \item
    For every clause $C = \ell_i \vee \ell_j \vee \ell_k$ of $F$ we introduce the gadget $D_C$ shown in Fig.~\ref{fig:Disjoint-QK2-clause-gadget} which contains one specified 
    vertices $k_{C, 1}$.
Furthermore, we connect $D_C$ to the gadgets $D_i,D_j,D_k$ with three arcs $(k_{C,1}, \lambda_i)$, $(k_{C,1}, \lambda_j)$ and
    $(k_{C,1}, \lambda_k)$,
    where   $\lambda_i = \texttt{t}_i$ (resp. $\lambda_j = \texttt{t}_j$  \& $\lambda_k = \texttt{t}_k$) 
    if $\ell_i$ (resp. $\ell_j$  \& $\ell_k$) is a positive literal, and $\lambda_i = \texttt{f}_i$ (resp. $\lambda_j = \texttt{f}_j$  \& $\lambda_k = \texttt{f}_k$) 
    if $\ell_i$ (resp. $\ell_j$  \& $\ell_k$) is a negative literal.
    See Fig.\ref{fig:Disjoint-QK2-cnx-gadget} for an example.
  \end{itemize}
	Note that for every clause $C$ of $F$, the digraph $D_C$ is the counterexample constructed 
	by Gutin et al.~\cite{Gutin_2004}). It has the important property that any two distinct
	vertices of $\{k_{C, i} \colon 1 \leq i \leq 7\}$ have a common out-neighbour in
	$\{k_{C, i} \colon 1 \leq i \leq 7\}$.

  It is clear that $|V| = 14m + 6n + 6$ and $|A| = 31m + 11n + 9$.
  Moreover, $D$ has maximum out-degree six (but it has unbounded in-degree, see vertex $b'$).
  We claim that the boolean formula $F$ is satisfiable if and only if 
  the digraph $D$ has two disjoint quasi-kernels.

  Suppose that the boolean formula $F$ is satisfiable and consider any satisfying 
  assignment $\varphi$.
  Construct two subsets $Q_1, Q_2 \subseteq V$ as follows.
  \begin{itemize}
    \item
    The elements of $Q_1$ are the following vertices: 
    the vertices $b$ and $b'$ from $D_0$, the vertex $A'_i$ from $D_i$ for every  variable $x_i \in X$, and the vertices
    $k_{C, 6}$, $s_{C, 1}$, $s_{C, 3}$ and $s_{C, 7}$ from $D_C$ for every clause $C$ of $F$.
    \item 
    The elements of $Q_2$ are the following vertices: 
    the vertices $c$ and $c'$ from $D_0$, the vertices $A''_i$ and $\texttt{t}_i$ from $D_i$ for every variable $x_i \in X$ with $\varphi(x_i) = \texttt{true}$, or the vertices $A''_i$ and $\texttt{f}_i$ from $D_i$ with $\varphi(x_i) = \texttt{false}$ and the vertices $k_{C, 7}$, $s_{C, 2}$ and $s_{C, 4}$ from $D_C$ for every clause $C$ of $F$.
  \end{itemize}
  It is a simple matter to check that $Q_1$ and $Q_2$ are disjoint and that both 
  $Q_1$ and $Q_2$ are independent subsets.
  Furthermore, we claim that $Q_1$ and $Q_2$ are two quasi-kernels of $D$.
  The claim is clear for $Q_1$.
  As for $Q_2$, it is enough to show that, for every clause $C$, the vertex $s_{C, 1}$ is 
  at distance at most two of some vertex in $Q_2$.
  Indeed, let $C = \ell_i \vee \ell_j \vee \ell_k$ be a clause where $\ell_i$, $\ell_j$ and 
  $\ell_k$ are positive or negative literals.
  Since $\varphi$ is a satisfying assignment, there exists one literal, say $\ell_i$,
  that evaluates to true in the clause $C$.
  Therefore, 
  if $\varphi(x_i) = \texttt{true}$ then $\texttt{t}_i \in Q_2$ 
  and $(k_{C, 1}, \texttt{t}_i) \in A$, 
  and 
  if $\varphi(x_i) = \texttt{false}$ then $\texttt{f}_i \in Q_2$ 
  and $(k_{C, 1}, \texttt{f}_i) \in A$.
  Conversely, suppose that there exist two disjoint quasi-kernels $Q_1$ and $Q_2$ in $D$.
  Define an assignment $\varphi$ for the boolean formula $F$ as follows:
  for $1 \leq i \leq n$, 
  if $\texttt{t}_i \in Q_2$ then set $\varphi(x_i) = \texttt{true}$;
  otherwise set $\varphi(x_i) = \texttt{false}$.
  Let us show that $\varphi$ is a satisfying assignment.
  We first observe that 
  $Q_1 \cap \{a,b,c\} \neq \varnothing$ and $Q_2 \cap \{a,b,c\} \neq \varnothing$.
  Then it follows that $a' \notin Q_1 \cup Q_2$ (by independence), and hence 
  $b' \in Q_1 \cup Q_2$.
  Without loss of generality, suppose $b' \in Q_1$.
  Therefore, by independence, 
  $\texttt{t}_i \notin Q_1$ and $\texttt{f}_i \notin Q_1$ for $1 \leq i \leq n$.
  
  We need the following claim.
  
  \begin{claim}
  \label{claim:2QK-1}
    We have
    $\{k_{C,1}, k_{C,2}, k_{C,3}, k_{C,5}\} \cap (Q_1 \cup Q_2) = \varnothing$ for every clause $C$ of $F$,
  \end{claim}
  \begin{proof}
  We only prove $k_{C,1} \notin Q_1 \cup Q_2$ 
  (the proof is similar for $k_{C,2} \notin Q_1 \cup Q_2$, $k_{C,3} \notin Q_1 \cup Q_2$ and $k_{C,5} \notin Q_1 \cup Q_2$.)
  Suppose, aiming at a contradiction, that $k_{C,1} \in Q_1 \cup Q_2$.
  Without loss of generality we may assume $k_{C,1} \in Q_1$ (the argument is symmetric if $k_{C,1} \in Q_2$).
  Then it follows that $\{s_{C,2}, s_{C,3}, s_{C,5}\} \subseteq Q_1$,
  and hence $\{s_{C,2}, s_{C,3}, s_{C,5}\} \cap Q_2 = \varnothing$.
  But, for any vertex $k_{C,i}$, $2 \leq i \leq 7$, we can easily check that either
  $d(s_{C,2}, k_{C,i}) > 2$, or $d(s_{C,3}, k_{C,i}) > 2$, or $d(s_{C,5}, k_{C,i}) > 2$.
  Hence $Q_2$ is not a quasi-kernel of $D$.
  This is the sought contradiction.
  \end{proof}

  \begin{claim}
    \label{claim:2QK-2}
    We have
    $s_{C,1} \in Q_1$ for every clause $C$ of $F$.
  \end{claim}

  \begin{proof}
    Suppose, aiming at a contradiction, that  $s_{C,1} \notin Q_1$.
    Combining Claim~\ref{claim:2QK-1} with $\texttt{t}_i \notin Q_1$ and $\texttt{f}_i \notin Q_1$ 
    for $1 \leq i \leq n$, we conclude that no vertex in $Q_1$ is at distance at most two from 
    $s_{C,1}$. Therefore, $Q_1$ is not a quasi-kernel of $D$. This is a contradiction.
  \end{proof}
  
  Let $C = \ell_i \vee \ell_j \vee \ell_k$ be a clause.
  According to Claim~\ref{claim:2QK-2}, we have $s_{C,1} \in Q_1$.
  Furthermore, according to Claim~\ref{claim:2QK-1}, 
  $\{k_{C,1}, k_{C,2}, k_{C,3}, k_{C,5}\} \cap Q_2 = \varnothing$.
  Then it follows that $\{\lambda_i, \lambda_j, \lambda_k\} \cap Q_2 \neq \varnothing$ where $\lambda_i = \texttt{t}_i$ (resp. $\lambda_j = \texttt{t}_j$  \& $\lambda_k = \texttt{t}_k$) 
  if $\ell_i$ (resp. $\ell_j$  \& $\ell_k$) is a positive literal, and $\lambda_i = \texttt{f}_i$ (resp. $\lambda_j = \texttt{f}_j$  \& $\lambda_k = \texttt{f}_k$) 
  if $\ell_i$ (resp. $\ell_j$  \& $\ell_k$) is a negative literal.
  Therefore $\varphi$ is a satisfying assignment.
\end{proof}

Recall that, whereas the small quasi-kernel conjecture is true for sink-free planar digraphs~\cite{kostochka_towards_2020},
no sink-free planar digraph without two disjoint quasi-kernels is known so far.
Moreover, the sink-free digraph constructed in Theorem~\ref{theorem:2QK-1} is obviously not planar
as it uses the counterexample constructed by Gutin et al.~\cite{Gutin_2004} that contains 
an orientation of $K_7$.
This raises the question of deciding whether every sink-free planar digraph has two disjoint 
quasi-kernels.
Although we have not been able to answer to this question, we show that is it \NP-complete 
to distinguish those sink-free planar digraphs that have three disjoint quasi-kernels 
from those that do not (the proof is deferred to Appendix due to the space constraints).

\begin{theorem}
  \label{theorem:3QK}
  Deciding if a digraph has three disjoint quasi-kernels is \NP-complete, 
	even for bounded degree planar digraphs.
\end{theorem}

\begin{proof}
  Given a planar graph with maximum degree four \textsc{3-Coloring} 
  asks to decide whether there exists a proper vertex coloring of $G$ with three colors
  (\emph{i.e.}, a labeling of the vertices with three colors such that no 
  two distinct vertices incident to a common edge have the same color).
   We reduce from \textsc{3-Coloring} which is known to be \NP-complete for planar graphs with 
  maximum degree four~\cite{dailey_uniqueness_1980}.

  Let $G = (V, E)$ be a planar graph with $n$ vertices, $m$ edges and maximum degree four. We denote by $v_1,v_2,\ldots,v_n$ the vertices of $G$.
  Without loss of generality, we assume that $G$ has no isolated vertex.
  We construct a digraph $D = (V', A)$ as follows.
  For every $1 \leq i \leq n$, we introduce $C_i$ an oriented cycle of length three which
  contains a specified vertex $w_i$.
  For every edge $v_i v_j \in E$, we connect the two gadgets $C_i$ and $C_j$ with 
  two arcs $(w_i, w_j)$ and $(w_j, w_i)$ to $D$.
  More formally, we have
  \begin{alignat*}{3}
  &V' & &= &&\{w_i \colon 1 \leq i \leq n\} \cup \{z_{i, j} \colon 1 \leq i \leq n \text{ and } 1 \leq j \leq 2\} \, ,\\
& A & &=&& \{(w_i, w_j), (w_j, w_i) \colon v_i v_j \in E\} \cup
  \{(w_i, z_{i, 1}), (z_{i, 1}, z_{i, 2}), (z_{i, 2}, w_i) \colon 1 \leq i \leq n\} \, .
  \end{alignat*}

  It is clear that $|V'| = 3n$, $|A| = 2m + 3n$ and that the digraph $D$ is planar.
  Moreover, $D$ has maximum in-degree five and maximum out-degree five.
  See Fig.~\ref{fig:3QK example} for an example.

  \begin{figure}[t!]
    \centering
    \begin{tikzpicture}
    [
      scale=.9,
      >=latex,
      thick,
      vertex/.style={shape=circle,draw=black},
      qk/.style={vertex,fill=black!25}
    ]
    \node [draw,rounded corners,fill=black!5] (G) at (.75,-1) {$G$};

    \node [vertex,ultra thick,label=west:$v_1$,fill=red!100] (v1) at (0,0) {};
    \node [vertex,ultra thick,label=east:$v_2$,fill=blue!100] (v2) at (1.5,0) {};
    \node [vertex,ultra thick,label=east:$v_3$,fill=green!100] (v3) at (1.5,1.5) {};
    \node [vertex,ultra thick,label=west:$v_4$,fill=blue!100] (v4) at (0,1.5) {};
    \node [vertex,ultra thick,label=west:$v_5$,fill=red!100] (v5) at (.75,2.75) {};

    \draw (v1) -- (v2);
    \draw (v1) -- (v3);
    \draw (v1) -- (v4);
    \draw (v2) -- (v3);
    \draw (v3) -- (v4);
    \draw (v3) -- (v5);
    \draw (v4) -- (v5);

    \begin{scope}[xshift=6.75cm,yshift=-1cm]
        \node [draw,rounded corners,fill=black!5] (D) at (1.5,-2.5) {$D$};

        \node [vertex,ultra thick,label=west:$w_1$,fill=red!100] (w1) at (0,0) {};
        \node [vertex,ultra thick,label=east:$w_2$,fill=blue!100] (w2) at (3,0) {};
        \node [vertex,ultra thick,label=east:$w_3$,fill=green!100] (w3) at (3,2.5) {};
        \node [vertex,ultra thick,label=west:$w_4$,fill=blue!100] (w4) at (0,2.5) {};
        \node [vertex,ultra thick,label=west:$w_5$,fill=red!100] (w5) at (1.5,4) {};

        \node [vertex,label=south:$z_{1,1}$,fill=blue!100] (z11) at (-1,-1.5) {};
        \node [vertex,label=south:$z_{1,2}$,fill=green!100] (z14) at (1,-1.5) {};

        \node [vertex,label=south:$z_{2,1}$,fill=green!100] (z21) at (2,-1.5) {};
        \node [vertex,label=south:$z_{2,2}$,fill=red!100] (z24) at (4,-1.5) {};

        \node [vertex,label=south:$z_{3,1}$,fill=red!100] (z31) at (4.5,1.5) {}; 
        \node [vertex,label=north:$z_{3,2}$,fill=blue!100] (z34) at (4.5,3.5) {};

        \node [vertex,label=south:$z_{4,1}$,fill=green!100] (z41) at (-1.5,1.5) {};
        \node [vertex,label=north:$z_{4,2}$,fill=red!100] (z44) at (-1.5,3.5) {};

        \node [vertex,label=north:$z_{5,1}$,fill=blue!100] (z51) at (2.5,5.5) {};
        \node [vertex,label=north:$z_{5,2}$,fill=green!100] (z54) at (.5,5.5) {};

        \draw [->] (w1) to [bend left=15] (w2);
        \draw [->] (w2) to [bend left=15] (w1);
        \draw [->] (w1) to [bend left=15] (w3);
        \draw [->] (w3) to [bend left=15] (w1);
        \draw [->] (w1) to [bend left=15] (w4);
        \draw [->] (w4) to [bend left=15] (w1);
        \draw [->] (w2) to [bend left=15] (w3);
        \draw [->] (w3) to [bend left=15] (w2);
        \draw [->] (w3) to [bend left=15] (w4);
        \draw [->] (w4) to [bend left=15] (w3);
        \draw [->] (w3) to [bend left=15] (w5);
        \draw [->] (w5) to [bend left=15] (w3);
        \draw [->] (w4) to [bend left=15] (w5);
        \draw [->] (w5) to [bend left=15] (w4);

        \draw [->] (w1) -- (z11);
        \draw [->] (z11) -- (z14);
        \draw [->] (z14) -- (w1);

        \draw [->] (w2) -- (z21);
        \draw [->] (z21) -- (z24);
        \draw [->] (z24) -- (w2);

        \draw [->] (w3) -- (z31);
        \draw [->] (z31) -- (z34);
        \draw [->] (z34) -- (w3);

        \draw [->] (w4) -- (z41);
        \draw [->] (z41) -- (z44);
        \draw [->] (z44) -- (w4);

        \draw [->] (w5) -- (z51);
        \draw [->] (z51) -- (z54);
        \draw [->] (z54) -- (w5);
    \end{scope}
  
  \end{tikzpicture}
    \caption{\label{fig:3QK example}%
      Example of the construction presented in the proof of Theorem~\ref{theorem:3QK}.
    }
  \end{figure}
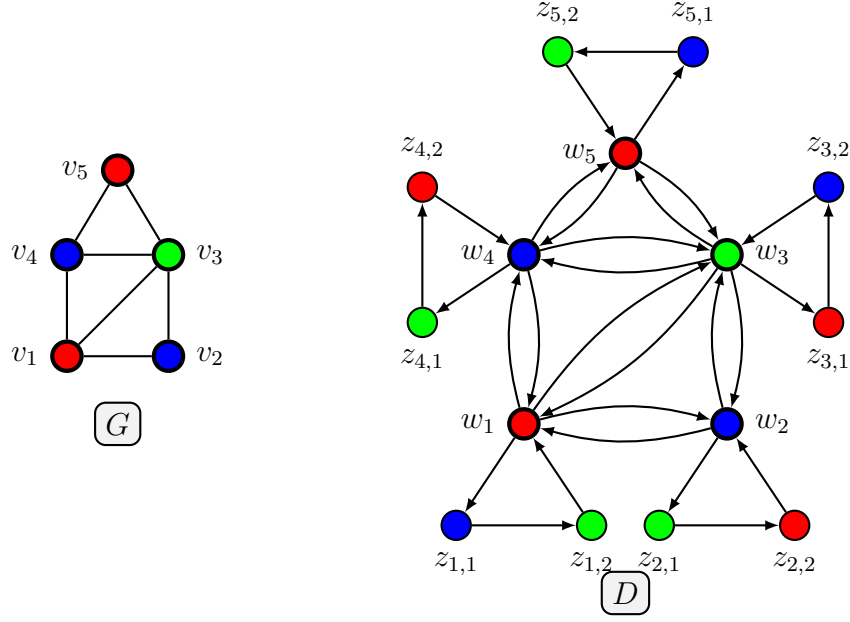

  We claim that $G$ has a proper $3$-coloring if and only if $D$ has three distinct quasi-kernels.

  Suppose first that $G$ has a proper $3$-coloring. 
  Let $C = \{c_1, c_2, c_3\}$ be the three colors.
  Consider the $3$-coloring of $D$ defined as follows: if $v_i$ is colored with color $c_1$ 
  (resp. $c_2$  \& $c_3$) in $G$, then 
  color $w_i$ with color $c_1$ (resp. $c_2$  \& $c_3$), 
  $z_{i,1}$ with color $c_2$ (resp. $c_3$  \& $c_1$), 
  and $z_{i,2}$ with color $c_3$ (resp. $c_1$  \& $c_2$) in $D$.  
  It is clear that the $3$-coloring of $D$ is proper.
  Moreover, for each color $c$ and every vertex $v$ there exist a directed path of length at most two from $v$ to a vertex colored with color $c$. 
  Then it follows that the $3$-coloring of $D$ induces three disjoint quasi-kernels in $D$.
  (Note that these three disjoint quasi-kernels actually form a partition of the vertices of $D$.)

  Conversely, suppose that the digraph $D$ has three disjoint quasi-kernels $Q_1$, $Q_2$ and $Q_3$.
  We have $w_i \in Q_1 \cup Q_2 \cup Q_3$ for every $1 \leq i \leq n$.
  Indeed, $\{z_{i,1}\} \cup N^+(z_{i,1}) \cup N^{+}(N^{+}(z_{i,1})) = \{z_{i,1}, z_{i,2}, w_i\}$.
  Thus  a quasi-kernel necesseraly contains one vertex in $\{z_{i,1}, z_{i,2}, w_i\}$, which implies that one of the three disjoint quasi-kernels countains $w_i$.
  Furthermore, each subset $Q_1$, $Q_2$ and $Q_3$ is independent because it is a quasi-kernel.
  Therefore, the three quasi-kernels $Q_1$, $Q_2$ and $Q_3$ induce a proper $3$-coloring of $G$.
\end{proof}

\section{Minimum size quasi-kernels}
\label{section:Minimum quasi-kernels}
In this section we address the complexity status of \QK and \MinQK for various digraph classes.

Our first result deals with orientations of trees.
Consider an orientation $T$ of a tree. Picking an arbitrary vertex as a root, some quantities related to the minimum size of a quasi-kernel admit a nice recursive scheme with their values for the subtrees obtained by deleting the root.
This leads to a natural programming approach for this class of digraph. 
More precisely these quantities are the following, where $T_x$ denotes the subtree rooted at a vertex $x$: 
\begin{itemize}
\item[$\bullet$] $\pi_0(x)$ is the minimum size of a quasi-kernel of $T_x$ containing $x$, 
\item[$\bullet$] $\pi_1(x)$ is the minimum size of a quasi-kernel $Q$ of $T_x$ such that the shorter directed path from $x$ to $Q$ is of length one, 
\item[$\bullet$] $\pi_2(x)$ is the minimum size of a quasi-kernel $Q$ of $T_x$ such that the shorter directed path from $x$ to $Q$ is of length two,
\item[$\bullet$] $\rho(x)$ is the minimum size of a subset $S$ of vertices of $T_x$ 
that is independent and such that for every vertex $v\in T_x\setminus \{x\}$, there is a directed path from $v$ to $S$ of length at most two.
\end{itemize}
Observe that $\rho(x)$ is the minimal size of something which is almost a quasi-kernel of $T_x$: 
the only difference with a quasi-kernel is that the existence of a directed path from $x$ to the considered set of length at most two is not required.

The following lemma is the key result for obtaining the dynamic programming algorithm. Checking each of the equation is straightforward. 

\begin{lemma}
\label{lemma_tree}
Let $T$ be an orientation of a tree and let $v$ be an arbitrary vertex. 
The following equations  hold:

\begin{align*}
\pi_0(v) & =1+  \sum_{u\in N^-(v)}\left(\sum_{t\in N^-(u)}\rho(t) + \sum _{t'\in N^+(u)\setminus \{v\}} \min_{i=0,1,2} \pi_i(t') \right) + \sum_{w\in N^+(v)} \min_{i=1,2} \pi_i(w)\, ,\\
\pi_1(v) &= \sum_{u\in N^-(v)}\rho(u) + \min_{w\in N^+(v)} \left(\pi_0(w)+ \sum_{w'\in N^+(v)\setminus \{w\}} \min_{i=0,1,2} \pi_i(w') \right )\, ,\\
\pi_2(v) &= \sum_{u\in N^-(v)} \min _{i=0,1,2} \pi_i(u) + \min_{w \in N^+(v)} \left ( \pi_1(w) + \sum_{w'\in N^+(v) \setminus \{w\} } \min_{i=1,2} \pi_i(w') \right)\, , \\
\rho(v) &= \min\left \{\pi_0(v), \pi_1(v), \pi_2(v), \sum_{u\in N^-(v)} \min_{i=0,1,2} \pi(u) + \sum _{w \in N^+(v)} \pi_2(w) \right \}\, ,
\end{align*}
with the classical convention that a sum over an empty set is $0$ and the minimum over an empty set is $+\infty$.

\end{lemma} 

\begin{proposition}
  \label{Poly}
  \MinQK is polynomial time solvable for orientations of trees.
\end{proposition}

\begin{proof}
  Let $T$ be an orientation of a tree. Let $v$ be an arbitrary vertex of $T$ that plays the role of a root (note that in general $v$ has both in- and out-neighbors). 
  According to Lemma~\ref{lemma_tree}, the  quantities $\pi_0(u)$, $\pi_1(u)$, $\pi_2(u)$ and $\rho(u)$ can be computed bottom-up for all vertices $u$, from the leaves to $v$. The minimum size of a quasi-kernel of $T$ is then $ \min\{\pi_0(v), \pi_1(v), \pi_2(v)\} $.
\end{proof}

The next theorem shows that there is not so much room for extending Proposition~\ref{Poly}.

\begin{theorem}
  \label{NPCcubic}
  \QK is \NP-complete, even for acyclic orientations of cubic graphs. 
\end{theorem}

\begin{proof}
  
  Given a boolean expression $F$ in conjunctive normal form where each clause is the disjunction of 
  three distinct literals and each literal occurs exactly twice among the clauses, 
  \textsc{(3,B2)-SAT} asks to decide whether $F$ is satisfiable. 
  We reduce from \textsc{(3,B2)-SAT} which is known to be 
  \NP-complete~\cite{DBLP:journals/eccc/ECCC-TR03-049}.

\begin{figure}[h]
  \centering
  \subfloat[Gadget $D_{x_i}$ for boolean variable $x_i$.]{%
     \begin{tikzpicture}
    [
      scale=.9,
      >=latex,
      vertex/.style={shape=circle,draw=black},
      sink/.style={vertex,fill=red!100}
    ]
    \node (D) at (1,4) {$D_{x_i}$};
    \node [vertex,ultra thick,label=east:$\neg x_i$] (1) at (2,7) {};
    \node [vertex,ultra thick,label=west:$x_i$] (2) at (0,7) {};
    \node [vertex] (3) at (1,8) {};
    \node [vertex,label=west:$d_{i,2}$] (4) at (0,8) {};
    \node [vertex] (12) at (1,9) {};
    \node [vertex,label=east:$d_{i,3}$] (5) at (2,8) {};
    \node [vertex] (7) at (0,9) {};
    \node [vertex] (8) at (2,9) {};
    \node [vertex,label=north:$d_{i,1}$] (9) at (1,10) {};
    \node [vertex,label=north:$d_{i,4}$] (10) at (1,6) {};
    \node [vertex] (11) at (1,5) {};
    \node [vertex] (13) at (1,7) {};
    \node [vertex] (14) at (2,6) {};
    \node [vertex] (15) at (0,6) {};
    \path [->] (1) edge node[left]{}  (5);
    \path [->] (3) edge node[left] {}(4);
    \path [->] (3) edge node[left]  {}(5);
    \path [->] (2) edge node[left] {}(4);
    \path [->] (4) edge node[left]  {}(7);
    \path [->] (5) edge node[left] {}(8);
    \path [->] (7) edge node[left] {}(9);    
    \path [->] (8) edge node[left] {}(9);  
    \path [->] (13) edge node[left] {}(3); 
    \path [->] (12) edge node[left] {}(7); 
    \path [->] (12) edge node[left] {}(8); 
    \path [->] (12) edge node[left] {}(9); 
    \path [->] (11) edge node[left] {}(15); 
    \path [->] (11) edge node[left] {}(14); 
    \path [->] (11) edge node[left] {}(10); 
    \path [->] (15) edge node[left] {}(10);
    \path [->] (14) edge node[left] {}(10);
    \path [->] (15) edge node[left] {}(13);
    \path [->] (14) edge node[left] {}(13);
\end{tikzpicture}
    \label{sub:NPC gadget var}
  }
  \qquad
  \subfloat[Gadget $D_{C_j}$ for clause $C_j$.]{%
    \begin{tikzpicture}
  [
    scale=.9,
    >=latex,
    thick,
    vertex/.style={shape=circle,draw=black},
    qk/.style={vertex,fill=red!100},
    sink/.style={vertex,fill=red!100}
  ]
    
  \node (Dj) at (0,0) {$D_{C_j}$};
  \node [vertex] (1') at (-30:1) {};
  \node [vertex] (1'bis) at (-30:2) {};
  \node [vertex,ultra thick,label=north east:$C_{j,2}$] (Cj2) at (30:1) {};
  \node [vertex] (3') at (90:1) {};
  \node [vertex] (3'bis) at (90:2) {};
  \node [vertex,ultra thick,label=north west:$C_{j,1}$] (Cj1) at (150:1) {};
  \node [vertex] (5') at (210:1) {};
  \node [vertex] (5'bis) at (210:2) {};
  \node [vertex,ultra thick,label=south:$C_{j,3}$] (Cj3) at (270:1) {};

  \node [vertex,label=180:$t_j$] (1'4) at (-30:3) {};
  \node [vertex] (1'2) at (-30:4) {};
  \node [vertex] (1'1) at (-10:3) {};
  \node [vertex] (1'3) at (-50:3) {};

  \node [vertex,label=south:$t'_j$] (3'4) at (90:3) {};
  \node [vertex] (3'2) at (90:4) {};
  \node [vertex] (3'1) at (110:3) {};
  \node [vertex] (3'3) at (70:3) {};

  \node [vertex,label=0:$t''_j$] (5'4) at (210:3) {};
  \node [vertex] (5'2) at (210:4) {};
  \node [vertex] (5'1) at (230:3) {};
  \node [vertex] (5'3) at (190:3) {};

  \path [->] (1'bis) edge node[left]{}  (1');
  \path [->] (3'bis) edge node[left]{}  (3');
  \path [->] (5'bis) edge node[left]{}  (5');
  \path [->] (1') edge node[left]{}  (Cj2);
  \path [->] (3') edge node[left]{}  (Cj2);
  \path [->] (3') edge node[left]{}  (Cj1);
  \path [->] (5') edge node[left]{}  (Cj1);
  \path [->] (5') edge node[left]{}  (Cj3);
  \path [->] (1') edge node[left]{}  (Cj3);

  \path [->] (1'1) edge node[left]{}  (1'bis);
  \path [->] (1'3) edge node[left]{}  (1'bis);
  \path [->] (1'1) edge node[left]{}  (1'4);
  \path [->] (1'3) edge node[left]{}  (1'4);
  \path [->] (1'2) edge node[left]{}  (1'1);
  \path [->] (1'2) edge node[left]{}  (1'4);
  \path [->] (1'2) edge node[left]{}  (1'3);

  \path [->] (3'1) edge node[left]{}  (3'bis);
  \path [->] (3'3) edge node[left]{}  (3'bis);
  \path [->] (3'1) edge node[left]{}  (3'4);
  \path [->] (3'3) edge node[left]{}  (3'4);
  \path [->] (3'2) edge node[left]{}  (3'1);
  \path [->] (3'2) edge node[left]{}  (3'4);
  \path [->] (3'2) edge node[left]{}  (3'3);

  \path [->] (5'1) edge node[left]{}  (5'bis);
  \path [->] (5'3) edge node[left]{}  (5'bis);
  \path [->] (5'1) edge node[left]{}  (5'4);
  \path [->] (5'3) edge node[left]{}  (5'4);
  \path [->] (5'2) edge node[left]{}  (5'1);
  \path [->] (5'2) edge node[left]{}  (5'4);
  \path [->] (5'2) edge node[left]{}  (5'3);

\end{tikzpicture}
    \label{sub:NPC gadget clause}
  }
  \caption{\label{fig:renoNPC gadgets}%
  }
\end{figure}

Consider an instance of \textsc{(3,B2)-SAT}. Denote by $x_1,x_2,\ldots, x_n$ its variables, 
and let $F=C_1 \vee C_2 \vee \dots \vee C_m$ be its formula. 
For every boolean variable $x_i$ occurring in $F$ we introduce a copy $D_{x_i}$ of the gadget shown in
Fig.~\ref{sub:NPC gadget var} which contains two specified vertices $x_i$ and $\neg x_i$.
Furthermore, for every clause $C_j$ of $F$ we introduce a copy $D_{C_j}$ of the gadget shown in 
Fig.~\ref{sub:NPC gadget clause} which contains three specified vertices $C_{j,1}$, $C_{j,2}$ and $C_{j,3}$.

If the literal $x_i$ (resp. $\neg x_i$) occurs in clause $C_j$ as the 
$k$-th literal ($k = 1, 2, 3$), we connect the specified vertex $x_i$ (resp. $\neg x_i$) in $D_{x_i}$
with the specified vertex $C_{j,k}$ in $D_{C_j}$ with a directed path of length two from $x_i$ to $C_{j,k}$.
(For an example, see Fig.~\ref{fig:NP cubic digraph example} 
where $C_j = x_{i_1} \vee \neg x_{i_2} \vee \neg x_{i_3}$.) 
Let $D$ denote the resulting digraph.
We observe that $D$ is an orientation of a cubic graph (since every literal occurs twice)
with $14n + 21m + 3m= 14n + 24m$  vertices.

\begin{figure}[ht!] 
  \centering
  \begin{tikzpicture}
    [
      scale=.9,
      >=latex,
      thick,
      vertex/.style={shape=circle,draw=black},
      qk/.style={vertex,fill=red!100},
      sink/.style={vertex,fill=red!100}
    ]
    \useasboundingbox (-5,-3) rectangle (5,10);

    \node (Dj) at (0,0) {$D_{C_j}$};
    \node [vertex] (1') at (-30:1) {};
    \node [vertex] (1'bis) at (-30:2) {};
    \node [vertex,ultra thick,label=north east:$C_{j,2}$,qk] (Cj2) at (30:1) {};
    \node [vertex] (3') at (90:1) {};
    \node [vertex] (3'bis) at (90:2) {};
    \node [vertex,ultra thick,label=north west:$C_{j,1}$] (Cj1) at (150:1) {};
    \node [vertex] (5') at (210:1) {};
    \node [vertex] (5'bis) at (210:2) {};
    \node [vertex,ultra thick,label=south:$C_{j,3}$,qk] (Cj3) at (270:1) {};
  
    \node [sink] (1'4) at (-30:2.9) {};
    \node [vertex] (1'2) at (-30:4) {};
    \node [vertex] (1'1) at (-10:3) {};
    \node [vertex] (1'3) at (-50:3) {};
  
    \node [sink] (3'4) at (90:2.9) {};
    \node [vertex] (3'2) at (90:4) {};
    \node [vertex] (3'1) at (110:3) {};
    \node [vertex] (3'3) at (70:3) {};
  
    \node [sink] (5'4) at (210:2.9) {};
    \node [vertex] (5'2) at (210:4) {};
    \node [vertex] (5'1) at (230:3) {};
    \node [vertex] (5'3) at (190:3) {};
  
    \path [->] (1'bis) edge node[left]{}  (1');
    \path [->] (3'bis) edge node[left]{}  (3');
    \path [->] (5'bis) edge node[left]{}  (5');
    \path [->] (1') edge node[left]{}  (Cj2);
    \path [->] (3') edge node[left]{}  (Cj2);
    \path [->] (3') edge node[left]{}  (Cj1);
    \path [->] (5') edge node[left]{}  (Cj1);
    \path [->] (5') edge node[left]{}  (Cj3);
    \path [->] (1') edge node[left]{}  (Cj3);
  
    \path [->] (1'1) edge node[left]{}  (1'bis);
    \path [->] (1'3) edge node[left]{}  (1'bis);
    \path [->] (1'1) edge node[left]{}  (1'4);
    \path [->] (1'3) edge node[left]{}  (1'4);
    \path [->] (1'2) edge node[left]{}  (1'1);
    \path [->] (1'2) edge node[left]{}  (1'4);
    \path [->] (1'2) edge node[left]{}  (1'3);
  
    \path [->] (3'1) edge node[left]{}  (3'bis);
    \path [->] (3'3) edge node[left]{}  (3'bis);
    \path [->] (3'1) edge node[left]{}  (3'4);
    \path [->] (3'3) edge node[left]{}  (3'4);
    \path [->] (3'2) edge node[left]{}  (3'1);
    \path [->] (3'2) edge node[left]{}  (3'4);
    \path [->] (3'2) edge node[left]{}  (3'3);
  
    \path [->] (5'1) edge node[left]{}  (5'bis);
    \path [->] (5'3) edge node[left]{}  (5'bis);
    \path [->] (5'1) edge node[left]{}  (5'4);
    \path [->] (5'3) edge node[left]{}  (5'4);
    \path [->] (5'2) edge node[left]{}  (5'1);
    \path [->] (5'2) edge node[left]{}  (5'4);
    \path [->] (5'2) edge node[left]{}  (5'3);

    \begin{scope}[xshift=-6cm,yshift=.5cm]
        \node (Di1) at (1,4.5) {$D_{x_{i_1}}$};
        \node [vertex,ultra thick,label=east:$\neg x_{i_1}$] (nxi1) at (2,7) {};
        \node [vertex,ultra thick,label=west:$x_{i_1}$,qk] (xi1) at (0,7) {};
        \node [vertex] (3) at (1,8) {};
        \node [vertex] (4) at (0,8) {};
        \node [vertex] (12) at (1,9) {};
        \node [vertex,qk] (5) at (2,8) {};
        \node [vertex] (7) at (0,9) {};
        \node [vertex] (8) at (2,9) {};
        \node [sink] (9) at (1,10) {};
        \node [sink] (10) at (1,6) {};
        \node [vertex] (11) at (1,5) {};
        \node [vertex] (13) at (1,7) {};
        \node [vertex] (14) at (2,6) {};
        \node [vertex] (15) at (0,6) {};
        \path [->] (nxi1) edge node[left]{}  (5);
        \path [->] (3) edge node[left] {}(4);
        \path [->] (3) edge node[left]  {}(5);
        \path [->] (xi1) edge node[left] {}(4);
        \path [->] (4) edge node[left]  {}(7);
        \path [->] (5) edge node[left] {}(8);
        \path [->] (7) edge node[left] {}(9);    
        \path [->] (8) edge node[left] {}(9);  
        \path [->] (13) edge node[left] {}(3); 
        \path [->] (12) edge node[left] {}(7); 
        \path [->] (12) edge node[left] {}(8); 
        \path [->] (12) edge node[left] {}(9); 
        \path [->] (11) edge node[left] {}(15); 
        \path [->] (11) edge node[left] {}(14); 
        \path [->] (11) edge node[left] {}(10); 
        \path [->] (15) edge node[left] {}(10);
        \path [->] (14) edge node[left] {}(10);
        \path [->] (15) edge node[left] {}(13);
        \path [->] (14) edge node[left] {}(13);
    \end{scope}

    \begin{scope}[xshift=-1cm,yshift=.5cm]
        \node (Di2) at (1,4.5) {$D_{x_{i_2}}$};
        \node [vertex,ultra thick,label=east:$\neg x_{i_2}$] (nxi2) at (2,7) {};
        \node [vertex,ultra thick,label=west:$x_{i_2}$,qk] (xi2) at (0,7) {};
        \node [vertex] (3) at (1,8) {};
        \node [vertex] (4) at (0,8) {};
        \node [vertex] (12) at (1,9) {};
        \node [vertex,qk] (5) at (2,8) {};
        \node [vertex] (7) at (0,9) {};
        \node [vertex] (8) at (2,9) {};
        \node [sink] (9) at (1,10) {};
        \node [sink] (10) at (1,6) {};
        \node [vertex] (11) at (1,5) {};
        \node [vertex] (13) at (1,7) {};
        \node [vertex] (14) at (2,6) {};
        \node [vertex] (15) at (0,6) {};
        \path [->] (nxi2) edge node[left]{}  (5);
        \path [->] (3) edge node[left] {}(4);
        \path [->] (3) edge node[left]  {}(5);
        \path [->] (xi2) edge node[left] {}(4);
        \path [->] (4) edge node[left]  {}(7);
        \path [->] (5) edge node[left] {}(8);
        \path [->] (7) edge node[left] {}(9);    
        \path [->] (8) edge node[left] {}(9);  
        \path [->] (13) edge node[left] {}(3); 
        \path [->] (12) edge node[left] {}(7); 
        \path [->] (12) edge node[left] {}(8); 
        \path [->] (12) edge node[left] {}(9); 
        \path [->] (11) edge node[left] {}(15); 
        \path [->] (11) edge node[left] {}(14); 
        \path [->] (11) edge node[left] {}(10); 
        \path [->] (15) edge node[left] {}(10);
        \path [->] (14) edge node[left] {}(10);
        \path [->] (15) edge node[left] {}(13);
        \path [->] (14) edge node[left] {}(13);
    \end{scope}

    \begin{scope}[xshift=4cm,yshift=.5cm]
        \node (Di3) at (1,4.5) {$D_{x_{i_3}}$};
        \node [vertex,ultra thick,label=east:$\neg x_{i_3}$,qk] (nxi3) at (2,7) {};
        \node [vertex,ultra thick,label=west:$x_{i_3}$] (xi3) at (0,7) {};
        \node [vertex] (3) at (1,8) {};
        \node [vertex,qk] (4) at (0,8) {};
        \node [vertex] (12) at (1,9) {};
        \node [vertex] (5) at (2,8) {};
        \node [vertex] (7) at (0,9) {};
        \node [vertex] (8) at (2,9) {};
        \node [sink] (9) at (1,10) {};
        \node [sink] (10) at (1,6) {};
        \node [vertex] (11) at (1,5) {};
        \node [vertex] (13) at (1,7) {};
        \node [vertex] (14) at (2,6) {};
        \node [vertex] (15) at (0,6) {};
        \path [->] (nxi3) edge node[left]{}  (5);
        \path [->] (3) edge node[left] {}(4);
        \path [->] (3) edge node[left]  {}(5);
        \path [->] (xi3) edge node[left] {}(4);
        \path [->] (4) edge node[left]  {}(7);
        \path [->] (5) edge node[left] {}(8);
        \path [->] (7) edge node[left] {}(9);    
        \path [->] (8) edge node[left] {}(9);  
        \path [->] (13) edge node[left] {}(3); 
        \path [->] (12) edge node[left] {}(7); 
        \path [->] (12) edge node[left] {}(8); 
        \path [->] (12) edge node[left] {}(9); 
        \path [->] (11) edge node[left] {}(15); 
        \path [->] (11) edge node[left] {}(14); 
        \path [->] (11) edge node[left] {}(10); 
        \path [->] (15) edge node[left] {}(10);
        \path [->] (14) edge node[left] {}(10);
        \path [->] (15) edge node[left] {}(13);
        \path [->] (14) edge node[left] {}(13);
    \end{scope}

    \node [vertex] (Cj1 to xi1) at (-7,4.5) {};
    \draw [->] (Cj1) to [bend left=30] (Cj1 to xi1);
    \draw [->] (Cj1 to xi1) to [bend left=10] (xi1);

    \node [vertex] (Cj2 to nxi2) at (3,4.5) {};
    \draw [->] (Cj2) to [bend right=53] (Cj2 to nxi2);
    \draw [->] (Cj2 to nxi2) to [bend right=10] (nxi2);

    \node [vertex] (Cj3 to nxi3) at (7,4.5) {};
    \draw [->] (Cj3) .. controls ($ (Cj3) +(5,-6)$) and ($ (Cj3 to nxi3) +(0,-1)$) .. (Cj3 to nxi3);
    \draw [->] (Cj3 to nxi3) to [bend right=10] (nxi3);

  \end{tikzpicture}
  
  \caption{\label{fig:NP cubic digraph example}%
    Proof of Theorem \ref{NPCcubic}: connecting gadget $D_{C_j}$ to gadgets $D_{x_{i_1}}$, $D_{x_{i_2}}$ 
    and $D_{x_{i_3}}$ 
    for clause $C_j = x_{i_1} \vee \neg x_{i_2} \vee \neg x_{i_3}$.
    Shown here is the assignment 
    $\varphi(x_{i_1}) = \texttt{true}$, 
    $\varphi(x_{i_2}) = \texttt{true}$ and 
    $\varphi(x_{i_3}) = \texttt{false}$,
    and the clause $C_j$ is satisfied by its first literal.
  }
\end{figure}
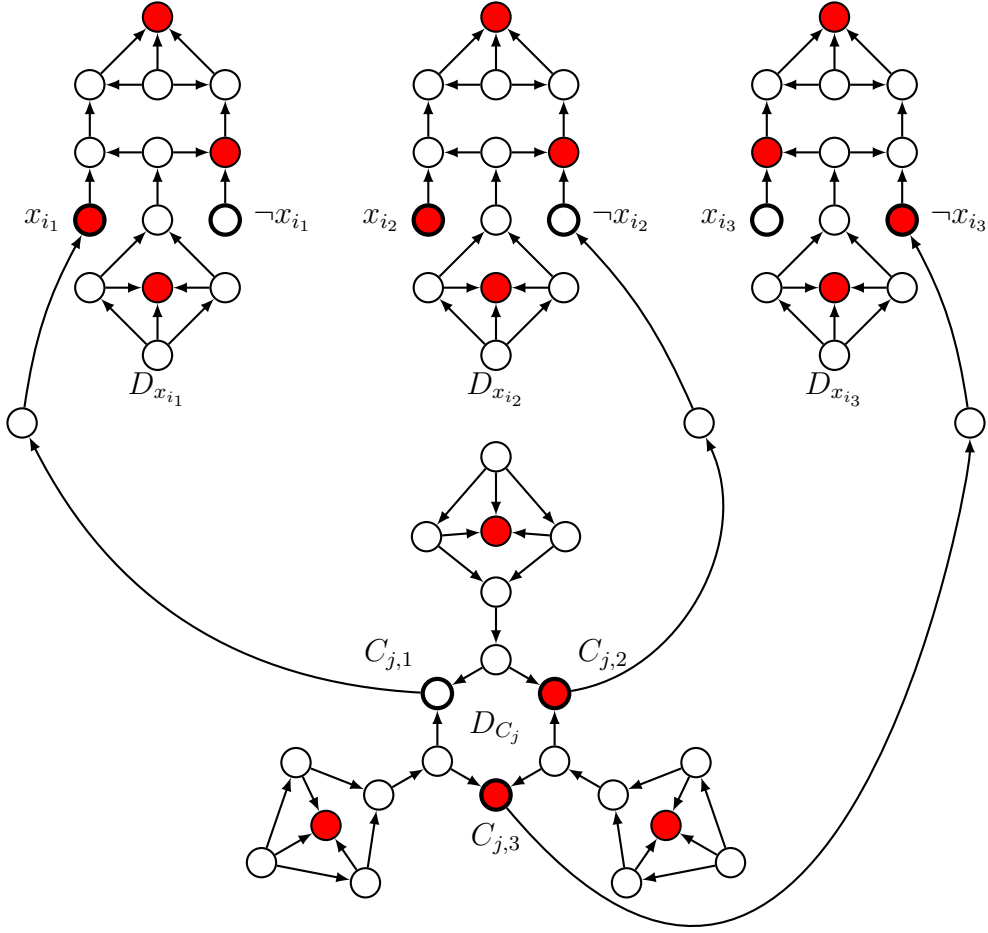

We claim that $F$ is satisfiable if and only if the digraph $D$ has a quasi-kernel of size 
$5m+4n$.

Suppose first that $F$ is satisfiable and consider a satisfying assignment $\varphi$. 
Construct a subset $Q$ of vertices of $D$ as follows.
\begin{itemize}
  \item 
  For $1 \leq j \leq m$, if the clause $C_j$ is satisfied by its first (resp. second \& third) literal,
  add the vertices $C_{j,2}$ and $C_{j,3}$ 
  (resp. $C_{j,1}$ and $C_{j,3}$ \& $C_{j,1}$ and $C_{j,2}$) to $Q$.
  In case the clause $C_j$ is satisfied by more than one literal, we choose one satisfying 
  literal arbitrarily.
  Furthermore, add the three vertices $t_{j}$, $t'_{j}$ and $t''_{j}$ to $Q$.
  \item 
  For $1 \leq i \leq n$, 
  if $\varphi(x_i) = \texttt{true}$ 
  (resp. $\varphi(x_i) = \texttt{false}$) 
  add the vertices $x_i$ and $d_{i,3}$ (resp. $\neg x_i$ and $d_{i,2}$) to $Q$.
  Furthermore, add the two vertices $d_{i,1}$ and $d_{i,4}$ to $Q$.
\end{itemize}
We check at once that $|Q| = 5m + 4n$.
It is now a simple matter to check that $Q$ is a quasi-kernel of $D$.

\begin{figure}[ht!]
  \centering
  \begin{tabular}{ccccc}
  \begin{tikzpicture}
    [
      scale=.9,
      >=latex,
      thick,
      vertex/.style={shape=circle,draw=black},
      qk/.style={vertex,fill=red!100},
      sink/.style={vertex,fill=red!100}
    ]
    \node (D) at (1,4.5) {$D_{x_i}$};
    \node [qk,ultra thick,label=east:$\neg x_i$] (1) at (2,7) {};
    \node [vertex,ultra thick,label=west:$x_i$] (2) at (0,7) {};
    \node [vertex] (3) at (1,8) {};
    \node [qk] (4) at (0,8) {};
    \node [vertex] (12) at (1,9) {};
    \node [vertex] (5) at (2,8) {};
    \node [vertex] (7) at (0,9) {};
    \node [vertex] (8) at (2,9) {};
    \node [sink] (9) at (1,10) {};
    \node [sink] (10) at (1,6) {};
    \node [vertex] (11) at (1,5) {};
    \node [vertex] (13) at (1,7) {};
    \node [vertex] (14) at (2,6) {};
    \node [vertex] (15) at (0,6) {};
    \path [->] (1) edge node[left]{}  (5);
    \path [->] (3) edge node[left] {}(4);
    \path [->] (3) edge node[left]  {}(5);
    \path [->] (2) edge node[left] {}(4);
    \path [->] (4) edge node[left]  {}(7);
    \path [->] (5) edge node[left] {}(8);
    \path [->] (7) edge node[left] {}(9);    
    \path [->] (8) edge node[left] {}(9);  
    \path [->] (13) edge node[left] {}(3); 
    \path [->] (12) edge node[left] {}(7); 
    \path [->] (12) edge node[left] {}(8); 
    \path [->] (12) edge node[left] {}(9); 
    \path [->] (11) edge node[left] {}(15); 
    \path [->] (11) edge node[left] {}(14); 
    \path [->] (11) edge node[left] {}(10); 
    \path [->] (15) edge node[left] {}(10);
    \path [->] (14) edge node[left] {}(10);
    \path [->] (15) edge node[left] {}(13);
    \path [->] (14) edge node[left] {}(13);
  \end{tikzpicture}
&\qquad&
\begin{tikzpicture}
  [
    scale=.9,
    >=latex,
    thick,
    vertex/.style={shape=circle,draw=black},
    qk/.style={vertex,fill=red!100},
    sink/.style={vertex,fill=red!100}
  ]
  \node (D) at (1,4.5) {$D_{x_i}$};
  \node [vertex,ultra thick,label=east:$\neg x_i$] (1) at (2,7) {};
  \node [qk,ultra thick,label=west:$x_i$] (2) at (0,7) {};
  \node [vertex] (3) at (1,8) {};
  \node [vertex] (4) at (0,8) {};
  \node [vertex] (12) at (1,9) {};
  \node [qk] (5) at (2,8) {};
  \node [vertex] (7) at (0,9) {};
  \node [vertex] (8) at (2,9) {};
  \node [sink] (9) at (1,10) {};
  \node [sink] (10) at (1,6) {};
  \node [vertex] (11) at (1,5) {};
  \node [vertex] (13) at (1,7) {};
  \node [vertex] (14) at (2,6) {};
  \node [vertex] (15) at (0,6) {};
  \path [->] (1) edge node[left]{}  (5);
  \path [->] (3) edge node[left] {}(4);
  \path [->] (3) edge node[left]  {}(5);
  \path [->] (2) edge node[left] {}(4);
  \path [->] (4) edge node[left]  {}(7);
  \path [->] (5) edge node[left] {}(8);
  \path [->] (7) edge node[left] {}(9);    
  \path [->] (8) edge node[left] {}(9);  
  \path [->] (13) edge node[left] {}(3); 
  \path [->] (12) edge node[left] {}(7); 
  \path [->] (12) edge node[left] {}(8); 
  \path [->] (12) edge node[left] {}(9); 
  \path [->] (11) edge node[left] {}(15); 
  \path [->] (11) edge node[left] {}(14); 
  \path [->] (11) edge node[left] {}(10); 
  \path [->] (15) edge node[left] {}(10);
  \path [->] (14) edge node[left] {}(10);
  \path [->] (15) edge node[left] {}(13);
  \path [->] (14) edge node[left] {}(13);
\end{tikzpicture}
&\qquad&
\begin{tikzpicture}
  [
    scale=.9,
    >=latex,
    thick,
    vertex/.style={shape=circle,draw=black},
    qk/.style={vertex,fill=red!100},
    sink/.style={vertex,fill=red!100}
  ]
  \node (D) at (1,4.5) {$D_{x_i}$};
  \node [vertex,ultra thick,label=east:$\neg x_i$] (1) at (2,7) {};
  \node [vertex,ultra thick,label=west:$x_i$] (2) at (0,7) {};
  \node [vertex] (3) at (1,8) {};
  \node [qk] (4) at (0,8) {};
  \node [vertex] (12) at (1,9) {};
  \node [qk] (5) at (2,8) {};
  \node [vertex] (7) at (0,9) {};
  \node [vertex] (8) at (2,9) {};
  \node [sink] (9) at (1,10) {};
  \node [sink] (10) at (1,6) {};
  \node [vertex] (11) at (1,5) {};
  \node [vertex] (13) at (1,7) {};
  \node [vertex] (14) at (2,6) {};
  \node [vertex] (15) at (0,6) {};
  \path [->] (1) edge node[left]{}  (5);
  \path [->] (3) edge node[left] {}(4);
  \path [->] (3) edge node[left]  {}(5);
  \path [->] (2) edge node[left] {}(4);
  \path [->] (4) edge node[left]  {}(7);
  \path [->] (5) edge node[left] {}(8);
  \path [->] (7) edge node[left] {}(9);    
  \path [->] (8) edge node[left] {}(9);  
  \path [->] (13) edge node[left] {}(3); 
  \path [->] (12) edge node[left] {}(7); 
  \path [->] (12) edge node[left] {}(8); 
  \path [->] (12) edge node[left] {}(9); 
  \path [->] (11) edge node[left] {}(15); 
  \path [->] (11) edge node[left] {}(14); 
  \path [->] (11) edge node[left] {}(10); 
  \path [->] (15) edge node[left] {}(10);
  \path [->] (14) edge node[left] {}(10);
  \path [->] (15) edge node[left] {}(13);
  \path [->] (14) edge node[left] {}(13);
\end{tikzpicture}
\end{tabular}
  \caption{\label{fig:NP cubic choice}%
    Proof of Theorem \ref{NPCcubic}: truth selection. 
  }
\end{figure}
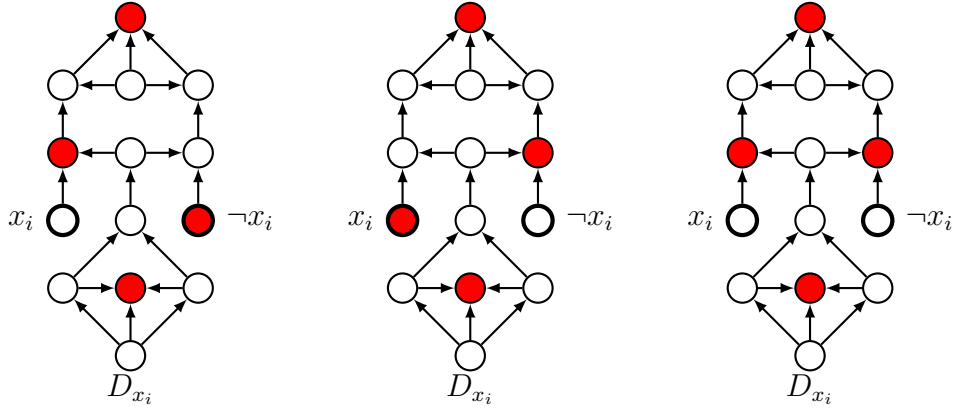

Conversely,
let $Q$ be a quasi-kernel of $D$ of size $5m+4n$.
We first observe that $Q$ contains at least five vertices of each gadget $D_{C_j}$.
We also observe that $Q$ contains at least four vertices of each gadget $D_{x_i}$.
Then it follows that $Q$ contains exactly five vertices of each gadget $D_{C_j}$,
and exactly $4$ vertices of each gadget $D_{x_i}$.
Note that we have actually shown that no mid-vertex of a path of length two
connecting some gadget $D_{C_j}$ to some gadget $D_{x_i}$ is in $Q$.
Furthermore, while we have a fair degree of flexibility in the way the vertices are selected, 
all sinks have to be in $Q$ (by definition). 
In particular, for each $j$, the three sinks $t_j$, $t'_j$ and $t''_j$ are in $Q$,
and for each $i$, the two sinks $d_{i,1}$ and $d_{j,4}$ are in $Q$ as well.
We now turn to defining a truth assignment $\varphi$ for the variable of $F$.
As $Q$ contains exactly four vertices in every gadget $D_{x_i}$, we are 
left to consider the three possibilities depicted in Fig.~\ref{fig:NP cubic choice},
where the grey vertices denote the vertices in $Q$.
The truth assignment $\varphi$ is defined as follows:
$\varphi(x_i) = \texttt{false}$ if and only if $\neg x_i \in Q$.
We claim that $\varphi$ is a satisfying assignment.
Indeed, consider any clause $C_j$.
Combining the observation that $t_j$, $t'_j$ and $t''_j$ are in $Q$ together with
the fact that $Q$ contains exactly five vertices of the gadget $D_{C_j}$,
we conclude that one of $C_{i,1}$, $C_{j,2}$ and $C_{j,3}$ is not in $Q$,
and that this vertex is necessarily absorbed by some vertex $x_i$ or $\neg x_i$ in $Q$.
(Note that $\varphi(x_i) = \texttt{true}$ if and only if $x_i \in Q$ yields another satisfying assignment
for the proposed construction.)
\end{proof} 

Assuming $\text{\FPT} \neq \text{\W[2]}$, 
our next result shows that one cannot confine the seemingly inevitable
combinatorial explosion of computational difficulty to an additive function of the 
size of the quasi-kernel, even for very restricted digraph classes.

\begin{theorem}
  \label{W2bip}
  \QK is \W\textnormal{[2]}-complete,
  even for acyclic orientations of bipartite graphs. 
\end{theorem}

\begin{proof}
  Membership to \W[2] is easy.
 
  Given $\mathcal{F}$ a family of sets over a universe $U$ and 
  a positive integer $k$, \textsc{Set Cover} consists in deciding if there exists a subfamily 
  $\mathcal{F}' \subseteq \mathcal{F}$ of size at most $k$ such that $\mathcal{F}'$ covers $U$.
   We prove hardness by reducing from \textsc{Set Cover} which is known to be 
  \W[2]-complete~\cite{DBLP:series/txcs/DowneyF13}.

  Let $\mathcal{F} = \{F_1, F_2, \dots, F_m\}$ be a family of sets over some universe 
  $U = \{u_1, u_2, \dots, u_n\}$ and $k$ be a positive integer. 
  Without loss of generality, we may assume $U = \bigcup_{1 \leq j \leq m} F_j$.
  We show how to produce a digraph $D = (V, A)$ such that $\mathcal{F}$ has a set cover of size at most $k$ 
  if and only if $D$ has a quasi-kernel of size at most $k' = k+1$.
  The digraph $D = (V, A)$ is defined as follows:
  \begin{align*}{}
    V &= \mathcal{F} \cup U \cup \{s, t\}\, , \\
    A &= \{(u_i, F_j) \colon F_j \in \mathcal{F} \text{ and } u_i \in F_j\} 
    \cup \{(F_j, s) \colon F_j \in \mathcal{F}\} 
    \cup \{(s,t)\}\, .
  \end{align*}
  It is clear that $|V| = m + n + 2$ and that $|A| = m + 1 + \sum_{1 \leq j \leq m} |F_j|$.

Suppose that there exists a subfamily 
$\mathcal{F}' \subseteq \mathcal{F}$ of size at most $k$ such that $\mathcal{F}'$ covers $U$.
It is clear that $\mathcal{F}' \cup \{t\}$ is a quasi-kernel of $D$ of size $k' = k+1$.

Conversely, suppose that there exists a quasi-kernel of $D$ of size at most $k' = k+1$.
Observe that $t \in Q$ since $t$ is a sink, and $s\notin Q$ (by independence). 
Among these quasi-kernels, choose one $Q$ that minimizes $|Q \cap U|$.
We show that $Q \cap U = \varnothing$. 
Indeed, suppose, aiming at a contradiction, that $Q \cap U \neq \varnothing$ and let $u_i \in Q \cap U$.
Furthermore, let $F_j \in N^+(u_i)$ and $U' = N^-(F_j)$.
Since $s \notin Q$ we have
$Q' = \left(Q \setminus U'\right) \cup \{F_j\}$
is a quasi-kernel of $D$ of size at most $k'$ with $|Q' \cap U| < |Q \cap U|$.
This contradicts our assumption, and hence $Q \cap U = \varnothing$.
Then it follows $N^+(u_i) \cap Q \neq \varnothing$ for every $u_i \in U$, and hence
$Q \cap \mathcal{F}$ yields a set cover of size $k' - 1 = k$.
\end{proof}
We remind that a kernel is a quasi-kernel. Actually we have slightly more: a kernel is an inclusion-wise maximal quasi-kernel.  Inclusion-wise minimal quasi-kernels are easy to find with a greedy algorithm. Though, finding a minimum-size quasi-kernel included in a kernel is hard as shown by the following result, whose proof is identical to the one of Theorem~\ref{W2bip}  ($\mathcal{F} \cup \{t\}$ is actually a kernel of 
the digraph $D$).

\begin{theorem}
  \label{W2bipK}
  Let $D = (V, A)$ be an acyclic orientation of a bipartite graph, $K \subseteq V$ be a kernel of $D$ and 
  $k$ be a positive integer. 
  Deciding whether there exists a quasi-kernel $Q$ included in $K$ of size $k$ is 
  \W\textnormal{[2]}-complete.
\end{theorem}

Raz and Safra~\cite{DBLP:conf/stoc/RazS97} have shown that \textsc{Set Cover} cannot be approximated in 
polynomial time to within a factor of $c \ln(|U|)$ for some constant $c$ unless 
$\P = \NP$.
Moreover, they built an instance of \textsc{Set Cover} 
where the number of subsets is a polynomial of the universe size.
Therefore, the construction used in the proof of Theorem~\ref{W2bip} allows us to 
state the following inapproximability result.

\begin{theorem}
  \label{th:inapproximability}
  \MinQK
  cannot be approximated in polynomial time to within a factor of $c \ln(|V|)$ for some constant $c$ 
  unless $\P = \NP$,
  even for acyclic orientations of bipartite graphs. 
\end{theorem}

Our last result focuses on even more restricted classes of digraph, namely acyclic digraphs with bounded degrees. We need a preliminary lemma which we state for general digraphs.

\begin{lemma}
  \label{lemma:QK in APX}
  \QK belongs to \APX\ for digraphs with fixed maximum in-degrees.
\end{lemma}

\begin{proof}
  Let $D = (V, A)$ be a digraph and $Q \subseteq V$ be a quasi-kernel. 
  It is clear that $(d^2 + d + 1) |Q| \geq |V|$,
  where $d$ is the maximum in-degree of $D$.
  Then it follows that any polynomial time algorithm that computes a quasi-kernel (such as the algorithm proposed by 
  Chv\'atal and Lov\'asz~\cite{10.1007/BFb0066192})
  is a $(d^2 + d + 1)$-approximation algorithm.
\end{proof}

\begin{proposition}
  \label{proposition:APX-complete}
  \QK is \APX-complete, even for acyclic digraphs with maximum in-degree three.
\end{proposition}

\begin{proof}
  Membership to \APX\ for acyclic digraphs with fixed in-degrees
  follows from Lemma~\ref{lemma:QK in APX}.
  Specifically, \QK for acyclic digraphs with maximum in-degree three
  can be approximated in polynomial time to within a factor of $13$.

  To prove hardness, we $L$-reduce~\cite{DBLP:journals/jcss/PapadimitriouY91} 
  from \textsc{Vertex Cover} in cubic graphs 
  which is known to be \APX-complete~\cite{DBLP:conf/ciac/AlimontiK97}.
  Let $f$ be the following $L$-reduction from
  \textsc{Vertex Cover} in cubic graphs to \QK with maximum in-degree three.
  Given a cubic graph $G = (V, E)$ with $V=[n]$ and $m$ edges, 
  we construct a digraph $D = (V', A)$ as follows:
  \begin{alignat*}{3}
    &V'&&= \{w_i, w'_i, w''_i \colon 1 \leq i \leq n\}
    \cup \{z_e, z'_e \colon e \in E\}\, , \\
    &A&&= \{(w_i, w'_i), (w'_i, w''_i) \colon 1 \leq i \leq n\} 
    \cup \{(z'_e, z_e), (z_e, w_i), (z_e, w_j) \colon e = ij \in E\}\, .
  \end{alignat*}
  Note that the vertices $w''_i$ are sinks in $D$.
  It is clear that $|V'| = 3n + 2m$, $|A| = 2n + 3m$ and, since $G$ is a cubic graph, 
  that every vertex has maximum in-degree three in $D$
  (we also observe that the maximum out-degree is two in $D$).
  See Fig.~\ref{fig:APX-complete} for an example.
  \begin{figure}[ht!]
    \centering
    \begin{tikzpicture}
    [
      scale=.9,
      >=latex,
      thick,
      vertex/.style={shape=circle,draw=black},
      qk/.style={vertex,fill=red!100},
      sink/.style={vertex,fill=red!100}
    ]
    \node [] (G) at (1,-1) {$G$};
    \node [vertex,label=south:$1$,qk] (v1) at (0,0) {};
    \node [vertex,label=south:$2$] (v2) at (2,0) {};
    \node [vertex,label=east:$3$,qk] (v3) at (1,1) {};
    \node [vertex,label=east:$4$,qk] (v4) at (1,2) {};
    \node [vertex,label=north:$5$] (v5) at (0,3) {};
    \node [vertex,label=north:$6$,qk] (v6) at (2,3) {};

    \draw [-] (v1) -- (v2);
    \draw [-] (v2) -- (v3);    
    \draw [-] (v3) -- (v1);
    \draw [-] (v4) -- (v5);
    \draw [-] (v5) -- (v6);    
    \draw [-] (v6) -- (v4);
    \draw [-] (v1) -- (v5);
    \draw [-] (v2) -- (v6);
    \draw [-] (v3) -- (v4);

    \begin{scope}[xshift=7.5cm,yshift=-1.5cm]
        \node [] (D) at (2,-2) {$D$};

        \node [vertex,label=south:$w_1$,qk] (w1) at (0,0) {};
        \node [vertex,label=south:$w_2$] (w2) at (4,0) {};
        \node [vertex,label=east:$w_3$,qk] (w3) at (2,2) {};
        \node [vertex,label=east:$w_4$,qk] (w4) at (2,4) {};
        \node [vertex,label=north:$w_5$] (w5) at (0,6) {};
        \node [vertex,label=north:$w_6$,qk] (w6) at (4,6) {};

        \node [vertex,label=south:$w'_1$] (w1a) at ($ (w1) +(-1,0)$) {};
        \node [sink,label=south:$w''_1$] (w1b) at ($ (w1) +(-2,0)$) {};
        \node [vertex,label=south:$w'_2$] (w2a) at ($ (w2) +(1,0)$) {};
        \node [sink,label=south:$w''_2$] (w2b) at ($ (w2) +(2,0)$) {};
        \node [vertex,label=north:$w'_3$] (w3a) at ($ (w3) +(-1,0)$) {};
        \node [sink,label=north:$w''_3$] (w3b) at ($ (w3) +(-2,0)$) {};
        \node [vertex,label=south:$w'_4$] (w4a) at ($ (w4) +(-1,0)$) {};
        \node [sink,label=south:$w''_4$] (w4b) at ($ (w4) +(-2,0)$) {};
        \node [vertex,label=north:$w'_5$] (w5a) at ($ (w5) +(-1,0)$) {};
        \node [sink,label=north:$w''_5$] (w5b) at ($ (w5) +(-2,0)$) {};
        \node [vertex,label=north:$w'_6$] (w6a) at ($ (w6) +(1,0)$) {};
        \node [sink,label=north:$w''_6$] (w6b) at ($ (w6) +(2,0)$) {};

        \node [vertex,label=east:$z_{13}$] (z13a) at ($ (w1) +(+1,+1)$) {};
        \path (z13a) ++(140:1cm) node [vertex,label=south west:$z'_{13}$] (z13b) {}; 
        \node [vertex,label=north:$z_{12}$] (z12a) at ($ (w1) +(+2,0)$) {};
        \path (z12a) ++(-90:1cm) node [vertex,label=east:$z'_{12}$] (z12b) {}; 
        \node [vertex,label=west:$z_{23}$] (z23a) at ($ (w2) +(-1,+1)$) {};
        \path (z23a) ++(40:1cm) node [vertex,label=south east:$z'_{23}$] (z23b) {}; 
        \node [vertex,label=west:$z_{34}$] (z34a) at ($ (w3) +(0,+1)$) {};
        \path (z34a) ++(0:1cm) node [vertex,label=east:$z'_{34}$] (z34b) {}; 
        \node [vertex,label=east:$z_{45}$] (z45a) at ($ (w5) +(+1,-1)$) {};
        \path (z45a) ++(220:1cm) node [vertex,label=north west:$z'_{45}$] (z45b) {}; 
        \node [vertex,label=south:$z_{56}$] (z56a) at ($ (w5) +(+2,0)$) {};
        \path (z56a) ++(90:1cm) node [vertex,label=east:$z'_{56}$] (z56b) {}; 
        \node [vertex,label=west:$z_{46}$] (z46a) at ($ (w6) +(-1,-1)$) {};
        \path (z46a) ++(-40:1cm) node [vertex,label=north east:$z'_{46}$] (z46b) {}; 
        \node [vertex,label=east:$z_{15}$] (z15a) at ($ (z34a) +(-3,0)$) {};
        \path (z15a) ++(180:1cm) node [vertex,label=north:$z'_{15}$] (z15b) {}; 
        \node [vertex,label=west:$z_{26}$] (z26a) at ($ (z34a) +(3,0)$) {};
        \path (z26a) ++(0:1cm) node [vertex,label=north:$z'_{26}$] (z26b) {}; 
        \draw [->] (w1) -- (w1a);
        \draw [->] (w1a) -- (w1b);
        \draw [->] (w2) -- (w2a);
        \draw [->] (w2a) -- (w2b);
        \draw [->] (w3) -- (w3a);
        \draw [->] (w3a) -- (w3b);
        \draw [->] (w4) -- (w4a);
        \draw [->] (w4a) -- (w4b);
        \draw [->] (w5) -- (w5a);
        \draw [->] (w5a) -- (w5b);
        \draw [->] (w6) -- (w6a);
        \draw [->] (w6a) -- (w6b);
        \draw [->] (z13b) -- (z13a);
        \draw [->] (z13a) -- (w1);
        \draw [->] (z13a) -- (w3);
        \draw [->] (z12b) -- (z12a);
        \draw [->] (z12a) -- (w1);
        \draw [->] (z12a) -- (w2);
        \draw [->] (z23b) -- (z23a);
        \draw [->] (z23a) -- (w2);
        \draw [->] (z23a) -- (w3);
        \draw [->] (z34b) -- (z34a);
        \draw [->] (z34a) -- (w3);
        \draw [->] (z34a) -- (w4);
        \draw [->] (z45b) -- (z45a);
        \draw [->] (z45a) -- (w4);
        \draw [->] (z45a) -- (w5);
        \draw [->] (z56b) -- (z56a);
        \draw [->] (z56a) -- (w5);
        \draw [->] (z56a) -- (w6);
        \draw [->] (z46b) -- (z46a);
        \draw [->] (z46a) -- (w4);
        \draw [->] (z46a) -- (w6);
        \draw [->] (z15b) -- (z15a);
        \draw [->] (z15a) to [bend right=15] (w1);
        \draw [->] (z15a) to [bend left=15] (w5);
        \draw [->] (z26b) -- (z26a);
        \draw [->] (z26a) to [bend left=15] (w2);
        \draw [->] (z26a) to [bend right=15] (w6);
    \end{scope}
  \end{tikzpicture}
  
    \caption{\label{fig:APX-complete}%
      Example of the construction presented in the proof of Theorem~\ref{proposition:APX-complete}.
    }
  \end{figure}
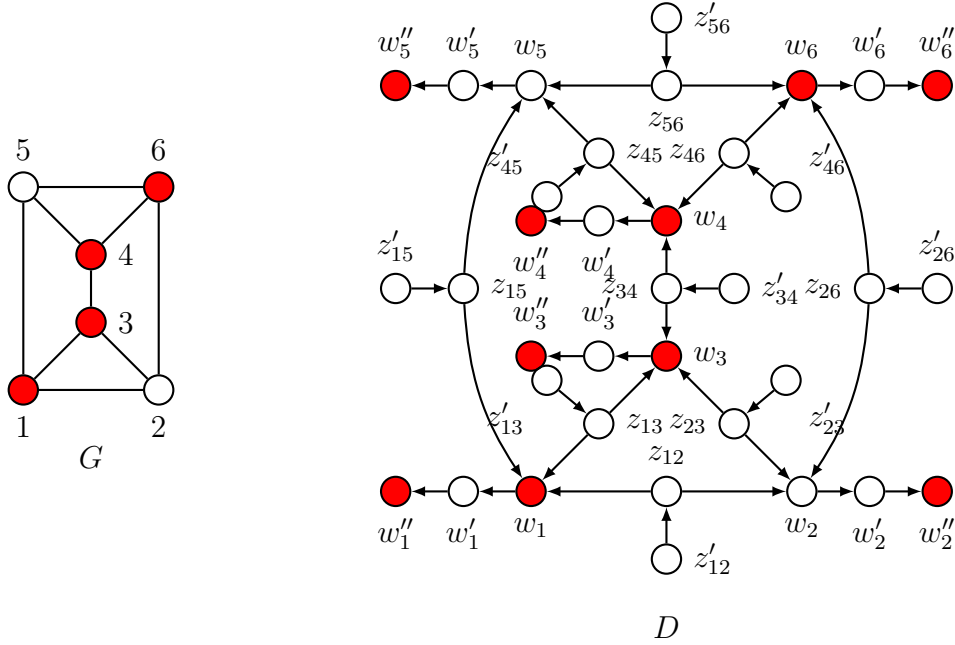

  Consider a quasi-kernel $Q \subseteq V'$ of $D = f(G)$. We claim that it can be transformed in polynomial time into a vertex cover $C \subseteq V$ of $G$ such that $|C| \leq |Q|$. To see this, observe first that $Q$ can be transformed 
  in polynomial time into a quasi-kernel $Q' \subseteq V'$ such that (i) $|Q'| \leq |Q|$ and
  (ii) $z'_e \notin Q'$ and $z_e \notin Q'$ for every $e \in E$.
  Indeed, repeated applications of the following two procedures enable us to achieve the claimed quasi-kernel.
  \begin{itemize}
    \item  
    Suppose that there exists $z'_e \in Q$ for some $e = ij \in E$.
    Then it follows that $z_e \notin Q$ (by independence).
    Furthermore, we have $w''_i \in Q$ and $w''_j \in Q$, and hence 
    $w'_i \notin Q$ and $w'_j \notin Q$.
    Therefore, $w_i \in Q$ or $w_j \in Q$ (possibly both).
    On account of the above remarks, $Q' = Q \setminus \{z'_e\}$ is a quasi-kernel of $D$
    and $|Q'| < |Q|$.
    \item
    Let $Z_i \subseteq Q$ stand for the set of vertices $z_e \in Q$, 
    where $e$ is an edge incident to the vertex $i$ in $G$.
    Suppose that there exists some set $Z_i \neq \varnothing$.
    Then it follows that $w_i \notin Q$ (by independence).
    Furthermore, we have $w''_i \in Q$, and hence $w'_i \notin Q$.
    On account of the above remarks, $Q' = \left(Q \setminus Z_i\right) \cup \{w_i\}$ 
    is a quasi-kernel of $D$ and $|Q'| \leq |Q|$.
  \end{itemize}
  From such a $Q'$, construct then a vertex cover $C \subseteq V$ of $G$ as follows:
  for $1 \leq i \leq n$, add the vertex $i$ to $C$ if $w_i \in Q'$. By construction, $C$ is a vertex cover of $G$
  of size $|C| = |Q'| - |V|$

  Finally, it is easy to see that from a vertex cover $C \in V$ of $G$ we can construct a 
  quasi-kernel $Q \subseteq V'$ of $D = f(G)$ of size exactly $|C| + |V|$:
  for every $1 \leq i \leq n$, add $w''_i$ to $Q$ and add $w_i$ to $Q$ if $i \in C$.
  Since $G$ is a cubic graph, we have $|C| \geq |V|/4$, and hence 
  $|Q| = |C| + |V| \leq |C| + 4|C| = 5|C|$.

  Thus $\opt(f(G)) \leq 5 \opt(G)$ and we have shown that $f$ is an $L$-reduction with parameters 
  $\alpha = 5$ and $\beta = 1$.
\end{proof}

\section{Concluding remarks}

In 2001 Gutin et al.~\cite{Gutin_2001} conjectured that every sink-free digraph has two disjoint quasi-kernel
(this stronger conjecture implies the small quasi-kernel conjecture~\cite{conj-qk}). 
Whereas this conjecture has been disproved by the same authors in 2004~\cite{Gutin_2004}, the clique number may play a role in this context.
Indeed, on the one hand, the key element on the counterexample constructed by Gutin et al.~\cite{Gutin_2004} is the presence of a $K_7$. 
On the other hand, the small quasi-kernel conjecture is true for sink-free orientations of $4$-colorable graphs~\cite{kostochka_towards_2020}, which have no $K_5$. 
This raises the question on whether every sink-free $K_5$-free digraph has two disjoint quasi-kernels, and, more generally, on how disjoint 
quasi-kernels and the clique number relate.

\bibliographystyle{amsplain}
\bibliography{quasi-kernel}

\providecommand{\bysame}{\leavevmode\hbox to3em{\hrulefill}\thinspace}
\providecommand{\MR}{\relax\ifhmode\unskip\space\fi MR }
\providecommand{\MRhref}[2]{%
  \href{http://www.ams.org/mathscinet-getitem?mr=#1}{#2}
}
\providecommand{\href}[2]{#2}
\begin{thebibliography}{10}

\bibitem{DBLP:conf/ciac/AlimontiK97}
P.~Alimonti and V.~Kann, \emph{Hardness of approximating problems on cubic
  graphs}, Algorithms and Complexity, Third Italian Conference, {CIAC} '97,
  Rome, Italy, March 12-14, 1997, Proceedings (G.C. Bongiovanni, D.P. Bovet,
  and G.~Di Battista, eds.), Lecture Notes in Computer Science, vol. 1203,
  Springer, 1997, pp.~288--298.

\bibitem{DBLP:journals/eccc/ECCC-TR03-049}
P.~Berman, M.~Karpinski, and A.D. Scott, \emph{Approximation hardness of short
  symmetric instances of {MAX-3SAT}}, Electron. Colloquium Comput. Complex.
  \textbf{049} (2003).

\bibitem{DBLP:journals/jgt/Bondy03}
J.A. Bondy, \emph{Short proofs of classical theorems}, J. Graph Theory
  \textbf{44} (2003), no.~3, 159--165.

\bibitem{DBLP:journals/dm/BorosG06}
E.~Boros and V.~Gurvich, \emph{Perfect graphs, kernels, and cores of
  cooperative games}, Discrete Math. \textbf{306} (2006), no.~19-20,
  2336--2354.

\bibitem{10.1007/BFb0066192}
V.~Chv{\'a}tal and L.~Lov{\'a}sz, \emph{Every directed graph has a
  semi-kernel}, Hypergraph Seminar (Berlin, Heidelberg) (C.~Berge and
  D.~Ray-Chaudhuri, eds.), Springer Berlin Heidelberg, 1974, p.~175.

\bibitem{chvatal_computational_1973}
V.~Chvátal, \emph{On the computational complexity of finding a kernel}, Tech.
  Report CRM300, Centre de Recherches Mathématiques, Université de Montréal,
  1973.

\bibitem{CROITORU2015863}
C.~Croitoru, \emph{A note on quasi-kernels in digraphs}, Information Processing
  Letters \textbf{115} (2015), no.~11, 863--865.

\bibitem{dailey_uniqueness_1980}
D.P. Dailey, \emph{Uniqueness of colorability and colorability of planar
  4-regular graphs are {NP}-complete}, Discrete Math. \textbf{30} (1980),
  no.~3, 289--293.

\bibitem{DBLP:series/txcs/DowneyF13}
R.G. Downey and M.R. Fellows, \emph{Fundamentals of parameterized complexity},
  Texts in Computer Science, Springer, 2013.

\bibitem{DBLP:conf/approx/DyerGGJ00}
M.E. Dyer, L.A. Goldberg, C.S. Greenhill, and M.~Jerrum, \emph{On the relative
  complexity of approximate counting problems}, Approximation Algorithms for
  Combinatorial Optimization, Third International Workshop, {APPROX} 2000,
  Saarbr{\"{u}}cken, Germany, September 5-8, 2000, Proceedings (K.~Jansen and
  S.~Khuller, eds.), Lecture Notes in Computer Science, vol. 1913, Springer,
  2000, pp.~108--119.

\bibitem{conj-qk}
P.L. Erdős and L.A. Székely, \emph{Two conjectures on quasi-kernels, open
  problems no. 4. in fete of combinatorics and computer science}, Bolyai
  Society Mathematical Studies, 2010.

\bibitem{FRAENKEL1981257}
A.S. Fraenkel, \emph{Planar kernel and {G}rundy with $d \leq 3$, $d_\text{out}
  \leq 2$, $d_\text{in} \leq 2$ are {NP}-complete}, Discrete Appl. Math.
  \textbf{3} (1981), no.~4, 257--262.

\bibitem{galvin_list_1995}
F.~Galvin, \emph{The list chromatic index of a bipartite multigraph}, J.
  Combin. Theory Ser. B \textbf{63} (1995), 153--158.

\bibitem{Gutin_2004}
G.~Gutin, K.M. Koh, E.G. Tay, and A.~Yeo, \emph{On the number of quasi-kernels
  in digraphs}, Report Series 01-7 (2001).

\bibitem{Gutin_2001}
\bysame, \emph{On the number of quasi-kernels in digraphs}, J. Graph Theory
  \textbf{46} (2004), no.~1, 48--56.

\bibitem{heard_disjoint_2008}
S.~Heard and J.~Huang, \emph{Disjoint quasi-kernels in digraphs}, J. Graph
  Theory \textbf{58} (2008), no.~3, 251--260 (en).

\bibitem{igarashi_coalition_2017}
A.~Igarashi, \emph{Coalition formation in structured environments}, Proceedings
  of the 16th Conference on Autonomous Agents and Multiagent Systems, 2017,
  pp.~1836--1837.

\bibitem{DBLP:journals/dm/JacobM96}
H.~Jacob and H.~Meyniel, \emph{About quasi-kernels in a digraph}, Discrete
  Math. \textbf{154} (1996), no.~1-3, 279--280.

\bibitem{Kar72}
R.~Karp, \emph{Reducibility among combinatorial problems}, Complexity of
  Computer Computations (R.~Miller and J.~Thatcher, eds.), Plenum Press, 1972,
  pp.~85--103.

\bibitem{kostochka_towards_2020}
A.~Kostochka, R.~Luo, and S.~Shan, \emph{Towards the {Small} {Quasi}-{Kernel}
  {Conjecture}}, arXiv:2001.04003 [math] (2020).

\bibitem{DBLP:journals/jcss/PapadimitriouY91}
C.H. Papadimitriou and M.~Yannakakis, \emph{Optimization, approximation, and
  complexity classes}, J. Comput. Syst. Sci. \textbf{43} (1991), no.~3,
  425--440.

\bibitem{DBLP:conf/stoc/RazS97}
R.~Raz and S.~Safra, \emph{A sub-constant error-probability low-degree test,
  and a sub-constant error-probability {PCP} characterization of {NP}},
  Proceedings of the Twenty-Ninth Annual {ACM} Symposium on the Theory of
  Computing, El Paso, Texas, USA, May 4-6, 1997 (F.~Thomson Leighton and P.W.
  Shor, eds.), {ACM}, 1997, pp.~475--484.

\bibitem{vonneumann1947}
J.~von Neumann and O.~Morgenstern, \emph{Theory of games and economic
  behavior}, Princeton University Press, 1947.

\bibitem{kernel_SAT}
M.~Walicki and S.~Dyrkolbotn, \emph{Finding kernels or solving {SAT}}, J. of
  Discrete Algorithms \textbf{10} (2012), 146--164.

\end{thebibliography}

\end{document}